\documentclass[journal,peerreview,nodraftcls]{IEEEtran}
\usepackage{amsmath,bbm,bm,amssymb}
\usepackage{graphicx,psfrag,epsf,subfig}
\usepackage{booktabs}
\usepackage{multirow}
\usepackage{adjustbox}
\usepackage{wrapfig}
\usepackage{algorithmic,algorithm}
\usepackage{hyperref}

\graphicspath{{./image/}}
\usepackage{amsthm}
\let\LambdaOLD\Lambda
\renewcommand{\Lambda}{{\bm\LambdaOLD}}
\let\Psiold\Psi
\renewcommand{\Psi}{{\bm\Psiold}}
\let\Sigmaold\Sigma
\renewcommand{\Sigma}{{\bm\Sigmaold}}
\newcommand{\Tr}{\mathrm{Tr}~}

\newcommand{\normal}{{\mathcal N}}

\usepackage[many]{tcolorbox}    
\newtcolorbox{mybox}[1][]{
    tikznode boxed title,
    enhanced,
    arc=0mm,
    interior style={white},
    attach boxed title to top center= {yshift=-\tcboxedtitleheight/2},
    fonttitle=\bfseries,
    colbacktitle=white,coltitle=black,
    boxed title style={size=normal,colframe=white,boxrule=0pt},
    #1}

\newtheorem{theorem}{Theorem}
\newtheorem{lemma}[theorem]{Lemma}

  \usepackage{cite}

\newcommand{\citep}{\cite}
\newcommand{\citet}{\cite}

\begin{document}
\title{A Matrix--free Likelihood Method for Exploratory Factor Analysis of High-dimensional Gaussian Data}
\author{{Fan Dai, Somak Dutta and Ranjan Maitra}
  \IEEEcompsocitemizethanks{\IEEEcompsocthanksitem The authors are
    with Iowa State University. Email: \{fd43,somakd,maitra\}@iastate.edu.}
  \IEEEcompsocitemizethanks{This research was supported in part by the United States Department of Agriculture (USDA) National
Institute of Food and Agriculture (NIFA) Hatch project IOW03617. 
       The research of the third
      author was also supported in part by the National
      Institute of Biomedical Imaging and Bioengineering (NIBIB) 
of the National Institutes of Health (NIH) under Grant R21EB016212.
The content of this paper is however solely the responsibility of the
authors and does not represent the official views of the NIBIB, the
NIH, the NIFA or the USDA.}
\IEEEcompsocitemizethanks{A poster based on this research won the
  first author an award in the Second Midwest Statistical Machine Learning
  Colloquium in 2019. }
}
\markboth{
}%
{Dai \MakeLowercase{\textit{et al.}}:Factor Analysis of Gaussian Data}
\clearpage
\setcounter{page}{1}


\IEEEcompsoctitleabstractindextext{%

  \begin{abstract}

This paper proposes a novel profile likelihood method for estimating the covariance parameters in exploratory factor analysis of high-dimensional Gaussian datasets with fewer observations than number of variables. An implicitly restarted Lanczos algorithm and a limited-memory quasi-Newton method are implemented to develop a matrix-free framework for likelihood maximization. 
Simulation results show that our method is substantially faster than the expectation-maximization solution without sacrificing accuracy. Our method is applied to fit factor models on data from suicide attempters, suicide ideators and a control group.

\end{abstract}

}
\maketitle

\IEEEdisplaynotcompsoctitleabstractindextext

\section{Introduction} 
\label{sec:intro}
Factor analysis \citep{anderson2003introduction} is a multivariate
statistical technique that characterizes dependence among variables
using a small number of latent factors. Suppose that we have a
sample ${\bf Y}_1,{\bf Y}_2,\ldots,{\bf Y}_n$ from the $p$-variate Gaussian distribution
$\normal_p(\bm\mu,\bm\Sigma)$ with mean vector $\bm\mu$ and a
covariance matrix $\bm\Sigma$. We assume that $\bm{\Sigma} =
\bm{\Lambda}\bm{\Lambda}^{\top} + \bm{\Psi},$ where $\bm{\Lambda}$ is
a $p\times q$ matrix  of rank $q$ that describes the amount of
variance shared among the $p$ coordinates and $\bm{\Psi}$ is a
diagonal matrix with positive diagonal entries representing the unique
variance specific to each coordinate. Factor analysis of Gaussian data
for $p<n$ was first formalized by \citet{lawley1940} with efficient maximum
likelihood (ML) estimation methods proposed
by~\citet{joreskog,maxwell,mardia,anderson2003introduction} and
others. These methods however do not apply to datasets with $p > n$
that occur in applications such as the processing of microarray data \citep{sundberg2016exploratory}, sequencing data
\citep{leek2007,leek2014,Buettner2017}, analyzing the transcription factor activity profiles of gene regulatory networks using massive gene expression datasets \citep{genefa}, portfolio analysis in stock return \citep{arfa} and others \citep{trendafilov2011exploratory}. Necessary and sufficient conditions for the existence of MLE when $p>n$ have been obtained by \cite{robertson2007maximum}. In such cases, the available computer memory may be inadequate to store the sample covariance matrix $\bf S$ or to make multiple copies of the dataset needed during the computation. 

The expectation-maximization (EM) approach of \citet{rubin} can be
applied to datasets with $p>n$ but is computationally slow. So,
here we develop a profile likelihood method for high-dimensional
Gaussian data. Our method allows us to compute the gradient of the
profile likelihood function at negligible additional computational
cost and to check first-order optimality, guaranteeing high
accuracy. We develop a fast sophisticated computational
framework called FAD (Factor Analysis of Data) to compute ML estimates of $\Lambda$ and $\Psi$. Our framework is implemented in an R~\citep{R} package called {\sf fad}.

The remainder of this paper is organized as follows. Section
\ref{sec:meth} describes the factor model for Gaussian data and an ML
solution using the EM algorithm, and then proposes the profile
likelihood and FAD. The performance of FAD relative to EM is evaluated
in Section \ref{sec:sim}. Section \ref{sec:app} applies our
methodology on a functional magnetic resonance imaging (fMRI) dataset
related to suicidal behavior \citep{just}. Section \ref{sec:dis}
concludes with some discussion. A online supplement, with sections,
tables and figures referenced with the prefix ``S'', is available. 


\section{Methodology}\label{sec:meth}

\subsection{Background and Preliminaries}
Let $\bf Y$ be the $n\times p$ data matrix with ${\bf Y}_i$ as its $i$th row. Then, in the setup of
Section \ref{sec:intro}, the ML method profiles out $\bm\mu$ using the sample mean vector and then maximizes the log-likelihood,
\begin{equation}\label{eqn:loglikelihood}
\ell(\Lambda,\Psi) = -\frac n2\{p\log(2\pi) + \log\det \Sigma +\Tr \Sigma^{-1}\mathbf{S}\}
\end{equation}
where $\bar{\mathbf{Y}} = \mathbf{Y}^\top{\bf 1}/n$, $\mathbf{S} =
(\mathbf{Y} - \mathbf{1}\bar{\mathbf{Y}}^\top)^\top(\mathbf{Y} -
\mathbf{1}\bar{\mathbf{Y}}^\top)/n$, where $\bf 1$ is the vector of
1s. The matrix $\bf S$ is almost surely singular and has rank $n$ when $p > n$.
The factor model \eqref{eqn:loglikelihood} is not identifiable because
the matrices $\Lambda$ and $\Lambda \mathbf{Q}$ give rise to the same
likelihood for any orthogonal matrix $\mathbf{Q}.$ So, additional
constraints (see \citep{anderson2003introduction, mardia} for more
details) are imposed.

\subsubsection{EM Algorithms for parameter estimation}
The EM algorithm \citep{rubin,mclachlan} exploits the structure of the
factor covariance matrix by assuming $q$-variate standard normal
latent factors 
and writing the factor model as
$\mathbf{Y}_i = {\bm\mu} + \Lambda \mathbf{Z}_i + \bm\epsilon_i $
where $\bm\epsilon_i$'s are i.i.d
$\normal_p(\mathbf{0},\mathbf{\Psi})$ and $\mathbf{Z}_i$'s are
independent of $\bm\epsilon_i$'s. The EM algorithm is easily
developed, with analytical expressions for both the expectation
(E-step) and maximization (M-step) steps that can be speedily computed
(see Section \ref{sec:supp-em}).  

Although EM algorithms are guaranteed to increase the likelihood at
each iteration and converge to a local maximum, they are well-known
for their slow convergence. Further, these iterations run in a
$(pq+p)$-dimensional space that can be slow for very large
$p$. Accelerated variants \citep{liu,roland} show superior performance
in low-dimensional problems but come   with additional computational
overhead that dominates the gain in rate of convergence in  high
dimensions. EM algorithms also compromise on numerical accuracy by not
checking for first-order optimality to enhance speed. So, we next
develop a fast and accurate method for exploratory factor analysis (EFA) that is applicable in high dimensions.

\subsection{Profile likelihood for parameter estimation}\label{sec:proflik}
We start with the common and computationally useful identifiability
restriction on $\Lambda$ that constrains
$\bm\Gamma = \Lambda^\top\Psi^{-1}\Lambda$ to be diagonal with decreasing diagonal entries. This scale-invariant constraint is completely determined except for changes in sign in the columns of $\Lambda$. Under this constraint, $\Lambda$ can be profiled out for a given $\Psi$ as described in the following

\begin{lemma}\label{lemma:profileout}
Suppose that $\Psi$ is a given p.d. diagonal matrix. Suppose that the $q$ largest singular values of $\mathbf{W} = n^{-1/2}(\mathbf{Y} - \mathbf{1}\bar{\mathbf{Y}}^\top)\Psi^{-1/2}$ are $\sqrt{\theta_1} \geq \sqrt{\theta_2} \geq \cdots \geq \sqrt{\theta_q}$ and the corresponding $p$-dimensional right-singular vectors are the columns of $\mathbf{V}_q.$ Then the function $\Lambda \mapsto \ell(\Lambda,\Psi)$ is maximized at ${\hat\Lambda} = \Psi^{1/2}\mathbf{V}_q\bm\Delta,$ where $\bm\Delta$ is a $q\times q$ diagonal matrix with $i$th diagonal entry as $[\max(\theta_i-1,0)]^{1/2}.$ The profile log-likelihood equals, 
\begin{equation}\label{eqn:profilelikelihood}
\ell_p(\Psi) = c - \frac n2 \left\{\log\det\Psi + \Tr \Psi^{-1}\mathbf{S} + \sum_{i=1}^q(\log\theta_i - \theta_i + 1)\right\}
\end{equation}
where $c$ is a constant that depends only on $\mathbf{Y},$ $n,$ $p$ and $q$ but not on $\Psi.$ Furthermore, the gradient of $\ell_p(\Psi)$ is given by:
$\nabla\ell_p(\Psi) = -\frac n2~\mathrm{diag}({\hat{\Lambda}}{\hat{\Lambda}}^\top + \Psi - \mathbf{S}).$
\end{lemma}
\begin{proof}
See Section \ref{sec:supp-lemma}.
\end{proof}

The profile log-likelihood $\ell_p(\Psi)$ in
\eqref{eqn:profilelikelihood} depends on $\mathbf{Y}$ only through the
$q$ largest singular values of $\mathbf{W}.$ So, in order to compute
$\ell_p(\Psi)$ and $\nabla\ell_p(\Psi)$ we need to only compute the
$q$ largest singular values of $\mathbf{W}$ and the right singular vectors. For $q<<\min(n,p)$, as is usually the case, these largest singular values and singular vectors can be computed very fast using Lanczos algorithm.

Further constraints on $\mathbf{\Psi}$ (e.g. $\mathbf{\Psi}=\sigma^2\mathbf{I}_p,$ $\sigma^2 >0$) can be easily incorporated. Also, $\nabla\ell_p(\Psi)$ is available in closed form that enables us to check first-order optimality and ensure high accuracy.

Finally, $\ell_p(\Psi)$ is expressed in terms of $\mathbf{S}.$ However, ML estimators are
scale-equivariant, so we can estimate $\Lambda$ and $\Psi$ using the
correlation matrix and scale back to $\bf S$. A particular advantage
of using the sample correlation matrix is that $\ell_p(\Psi)$ needs to be optimized over a fixed bounded rectangle $(0,1)^p$ that does not depend on the data and is conceivably numerically robust.

\subsection{Matrix--free computations}
\subsubsection{A Lanczos algorithm for calculating partial singular
  values and vectors}\label{sec:lanczos}
In order to compute the profile likelihood and its gradient, we need the $q$ largest singular values and right singular vectors of $\mathbf{W}.$ We use the Lanczos algorithm \citep{baglama,dutta} with reorthogonalization and implicit restart. Suppose that $m = \max\{2q+1,20\}$ and that $\mathbf{f}_1 \in \mathbb{R}^n$ is any random vector with $\|\mathbf{f}_1\|=1.$ Let $\mathbf{g}_1 = \mathbf{W}^{\top}\mathbf{f}_1,$ $\mathbf{F}_1 = \mathbf{f}_1$ and $\mathbf{G}_1 = \mathbf{g}_1.$ For $j=1,\ldots,m$ let $\mathbf{r}_j = \mathbf{Wg}_j - \alpha_j\mathbf{f}_j,$ reorthogonalize $\mathbf{r}_j = \mathbf{r}_j - \mathbf{F}_j\mathbf{F}_j^\top\mathbf{r}_j$ and set $\beta_j = \|\mathbf{r}_j\|,$ and if $j < m,$ update $\mathbf{f}_{j+1} = \mathbf{r}_j/\beta_j,$ $\mathbf{F}_{j+1} = [\mathbf{F}_j,~ \mathbf{f}_{j+1}],$ $\mathbf{g}_{j+1} = \mathbf{W}^\top\mathbf{f}_{j+1}-\beta_j\mathbf{g}_j,$ reorthogonalize $\mathbf{g}_{j+1} = \mathbf{g}_{j+1}-\mathbf{G}_j\mathbf{G}_j^\top\mathbf{g}_{j+1},$ $\alpha_{j+1} = \|\mathbf{g}_{j+1}\|$, $\mathbf{g}_{j+1} = \mathbf{g}_{j+1}/\alpha_{j+1},$ and set $\mathbf{G}_{j+1}=[\mathbf{G}_j,~\mathbf{g}_{j+1}].$

Next, consider the $m\times m$ bidiagonal matrix $\mathbf{B}_m$ with
diagonal entries $\alpha_1,\alpha_2,\ldots,\alpha_m$ with $(j,j+1)$th
entry $\beta_j$ for $j = 1,2,\ldots,m-1$ and all other entries as 0. Now
suppose that $h_1\geq h_2\geq\cdots \geq h_m$ are the singular values
of $\mathbf{B}_m$ and that $\tilde{\mathbf{u}}_j$'s and
$\tilde{\mathbf{v}}_j$'s are the corresponding right and left singular
vectors, which can be computed via a Sturm sequencing algorithm
\citep{wilkinson}. Also, let $\mathbf{u}_j =
\mathbf{F}_m\tilde{\mathbf{u}}_j$ and $\mathbf{v}_j =
\mathbf{G}_m\tilde{\mathbf{v}}_j$ $(1\leq j \leq m)$.  Then it can be
shown that for all $j,$ $\mathbf{W}^\top\mathbf{u}_j =
h_j\mathbf{v}_j$ and $\mathbf{Wv}_j = h_j\mathbf{u}_j +
\tilde{v}_{j,m}\mathbf{r}_m,$ where $\tilde{v}_{j,m}$ is the last
entry of $\tilde{\mathbf{v}}_j.$ Because $\|\mathbf{r}_m\| = \beta_m$
and $h_1$ is approximately the largest singular value of $\mathbf{W},$
the algorithm stops if $\beta_m|\tilde{v}_{j,m}| \leq h_1\delta$ holds
for  $j=1,2,\ldots,q,$ where $\delta$ is some prespecified tolerance
level, and $h_1,h_2,\ldots,h_q$ and
$\mathbf{v}_1,\mathbf{v}_2,\ldots,\mathbf{v}_q$ are accurate
approximations of the $q$ largest singular values and corresponding
right singular vectors of $\mathbf{W}$. 

Convergence of the reorthogonalized Lanczos algorithm often
suffers from numerical instability that slows down convergence. To
resolve this instability, \citet{baglama} proposed restarting the
Lanczos algorithm, but instead of starting from scratch, they
initialized with the first $q$ singular vectors. To that end, let
$\mathbf{f}_{m+1} = \mathbf{r}_m/\beta_m$ and reset $\mathbf{F}_{q+1}
= [\mathbf{u}_1,\ldots,\mathbf{u_q},\mathbf{f}_{m+1}].$ Then for
$j=1,2,\ldots,q$, let $\rho_j=\beta_m\tilde{v}_{j,m}$, and reset
$\mathbf{r}_q =
\mathbf{W}^\top\mathbf{f}_{m+1}-\sum_{j=1}^q\rho_j\mathbf{v}_j,$
$\alpha_{q+1}=\|\mathbf{r}_{q}\|,$ $\mathbf{g}_{q+1} =
\mathbf{r}_q/\alpha_{q+1},$ and $\mathbf{G}_{q+1} =
       [\mathbf{v}_1,\ldots,\mathbf{v}_q,\mathbf{g}_{q+1}].$ Define
       $\gamma = \mathbf{f}_{m+1}^\top\mathbf{Wg}_{q+1}$ and
       $\mathbf{r}_{q+1} = \mathbf{Wg}_{q+1} - \gamma
       \mathbf{f}_{m+1}.$ For $j=1, 2,\ldots,m-q-1,$ compute
       $\beta_{q+j} = \|\mathbf{r}_{q+j}\|,$ $\mathbf{f}_{q+j+1} =
       \mathbf{r}_{q+j}/\beta_{q+j},$ $\mathbf{F}_{q+j+1} =
              [\mathbf{F}_{q+j},~\mathbf{f}_{q+j+1}],$
              $\mathbf{g}_{q+j+1} =
              (\mathbf{I}-\mathbf{G}_{q+j}\mathbf{G}_{q+j}^\top)\mathbf{W}^\top\mathbf{f}_{q+j+1},$
              $\alpha_{q+j+1} = \|\mathbf{g}_{q+j+1}\|,$
              $\mathbf{g}_{q+j+1} = \mathbf{g}_{q+j+1}/\alpha_{q+j+1}$
              and $\mathbf{r}_{q+j+1}=(\mathbf{I} -
              \mathbf{F}_{q+j+1}\mathbf{F}_{q+j+1}^\top)\mathbf{Wg}_{q+j+1}.$
              This yields a  matrix $\mathbf{B}_{m}$ with entries
              $b_{j,j} = h_j$ and $b_{j,q}=\rho_j$ for
              $j=1,2,\ldots,q,$ and $b_{i,i} = \alpha_i$ for $q+1\leq
              i \leq m$ and $b_{i,i+1} = \beta_{i}$ for $q+1\leq i
              \leq m-1$, and all other entries 0. The matrix $\mathbf{B}_m$
              is not bidiagonal but is still small-dimensioned
              matrix whose singular value decomposition can be
              calculated very fast. Convergence of the Lanczos
              algorithm can be checked as before. This restart step is repeated until all the $q$ largest singular values converge.

The only way that $\mathbf{W}$ enters this algorithm is through
matrix-vector products of the forms $\mathbf{W}\mathbf{g}$ and
$\mathbf{W}^\top\mathbf{f}$, both of which can be computed without explicitly storing $\mathbf{W}.$ Overall, this algorithm yields the $q$ largest singular values and vectors in $O(qnp)$
 computational cost using only $O(qp)$ additional memory, resulting in
  substantial gains 
 over the traditional methods \citep{joreskog,maxwell}. These
 traditional methods require a full eigenvalue decomposition of
 $\mathbf{W}^\top\mathbf{W}$ that is of $O(p^3)$ computational complexity and requires $O(p^2)$ memory
 storage space.

Having described a scalable algorithm for computing $\ell_p(\Psi)$ and $\nabla\ell_p(\Psi)$, we detail a numerical algorithm for computing the ML estimators.

\subsubsection{Numerical optimization of the profile log-likelihood}
On the correlation scale, $\psi_{ii}$'s lie between 0 and 1. Under
this box constraint, the \textsf{factanal} function in \textsf{R} and
\textsf{factoran} function in
\textsf{MATLAB}\textsuperscript{\textregistered}~employ the
limited-memory Broyden-Fletcher-Goldfarb-Shanno quasi-Newton algorithm
\citep{byrd} with box-constraints (L-BFGS-B) to obtain the ML
estimator of $\Psi.$ However, in high dimensions, the advantages of
the L-BFGS-B algorithm  are particularly prominent. Because Newton
methods require the search direction
$-\mathbf{H}(\Psi)^{-1}\nabla\ell_p(\Psi)$, where $\mathbf{H}(\Psi)$
is the $p\times p$ Hessian matrix of $\ell_p(\Psi)$, they are
computationally prohibitive in high dimensions in terms of storage and
numerical complexity. The quasi-Newton BFGS replaces the computation
of the exact search direction by an iterative approximation using the
already computed values of $\ell_p(\Psi)$ and $\nabla\ell_p(\Psi).$
The limited-memory implementation, moreover, uses only the last few
(typically less than 10) values of $\ell_p(\Psi)$ and
$\nabla\ell_p(\Psi)$ instead of using all the past values. Overall,
L-BFGS-B reduces the storage cost from $O(p^2)$ to $O(np)$ and the
computational complexity from $O(p^3)$ to $O(np)$. Interested readers
are referred to  \cite{byrd,Zhu1994Algorithm7L} for more details on the L-BFGS-B algorithm.


The L-BFGS-B algorithm requires  both $\ell_p$ and $\nabla\ell_p$ to
be computed at each iteration.
Because $\nabla \ell_p$ is available as a byproduct while computing
$\ell_p$ (see Section Sections \ref{sec:proflik} and
\ref{sec:lanczos}), we modify the implementation to jointly compute
both quantities with a single call to the Lanczos algorithm at each L-BFGS-B iteration. In comparison to \textsf{R}'s default implementation (\textsf{factanal}) that separately calls $\ell_p$ and $\nabla\ell_p$ in its optimization routines, this tweak halves the computation time.


\section{Performance evaluations} 
\label{sec:sim}
\subsection{Experimental setup}
\label{sec:sim-setup}
The performance of FAD was compared to EM using  
100 simulated 
datasets with true $q=3$ or $5$ and
$(n,p)\in\{(100,1000),(225,3375),(400,8000)\}$. For each setting, we simulated $\psi_{ii}\sim$ i.i.d $\mathcal{U}(0.2,0.8)$ and $\lambda_{ij}\sim$ i.i.d. $\mathcal{N}(0,1)$ and set $\bm\mu=\bm0$. We also evaluated performance with  $(n,p,q)\in\{(160,24547,2),(180,24547,2),(340,24547,4)\}$ to
match the settings of our application in Section \ref{sec:app}: the
true $\Psi, \Lambda,\bm\mu$ were set to be the ML estimates from 
that  dataset.


For the EM algorithm, $\bm{\Lambda}$ was initialized as the first $q$
principal components (PCs) of the scaled data matrix computed via the
the Lanczos algorithm while $\bm{\Psi}$ was started at $\mathbf{I}_p
-\mathrm{diag}(\bm{\Lambda}\bm{\Lambda}^{\top})$. FAD requires only $\bm{\Psi}$ to
be initialized, which was done in the same way as the EM.
We stopped FAD when the relative increase in $\ell_p(\Psi)$ was below
100$\epsilon_0$ and $\|\nabla\ell_p\|_\infty < \sqrt\epsilon_0$ where 
$\epsilon_0$ is the machine tolerance, which in our case was
approximately $2.2\times10^{-16}$. 
The EM algorithm was terminated if the relative change in
$\ell(\Lambda,\Psi)$ was less than $10^{-6}$ and
$\|\nabla\ell_p\|_\infty< \sqrt\epsilon_0 $, or if the number of
iterations reached 5000. Therefore, FAD and EM had comparable stopping
criteria.  For each simulated dataset, we fit
models with $k=1,2,\cdots,2q$ factors and chose the number of factors
by the Bayesian Information Criterion (BIC): $-2{\hat\ell}_k + pq\log
n$ \citep{schwarz1978}, where ${\hat\ell}_k$ is the maximum
log-likelihood value with $k$ factors.  All experiments were done
using R~\citep{R} on a workstation with Intel E5-2640 v3 CPU clocked @2.60 GHz and 64GB RAM. 

\subsection{Results} 
\label{sec:sim-res}
Because BIC always correctly picked $q$, 
we evaluated model fit for each method in terms of $\ell(\hat\Lambda,\hat\Psi),$
$d_{\hat{\mathbf{R}}} =
\|\hat{\mathbf{R}}-\mathbf{R}\|_F/\|\mathbf{R}\|_F$ and $d_{\hat{\bm
    \Gamma}} = \|\hat{\bm\Gamma}-{\bm\Gamma}\|_F/\|{\bm\Gamma}\|_F$
where $\hat{\mathbf{\Gamma} }
= \hat{\bm{\Lambda}}^{\top}\hat{\bm{\Psi}}^{-1}\hat{\bm{\Lambda}}$ and
$\bf R$ and $\hat{\bf R}$ are the correlation matrices corresponding to 
$\Sigma$ and $\hat{\Sigma} = \hat{\bm{\Lambda}}\hat{\bm{\Lambda}}^{\top}+\hat{\bm{\Psi}}$. 
\subsubsection{CPU time}
\begin{figure}[h]
\vspace{-0.2in}
  \centering
\mbox{
  \subfloat{\label{fig-time-hd}\includegraphics[width = 0.5\textwidth]{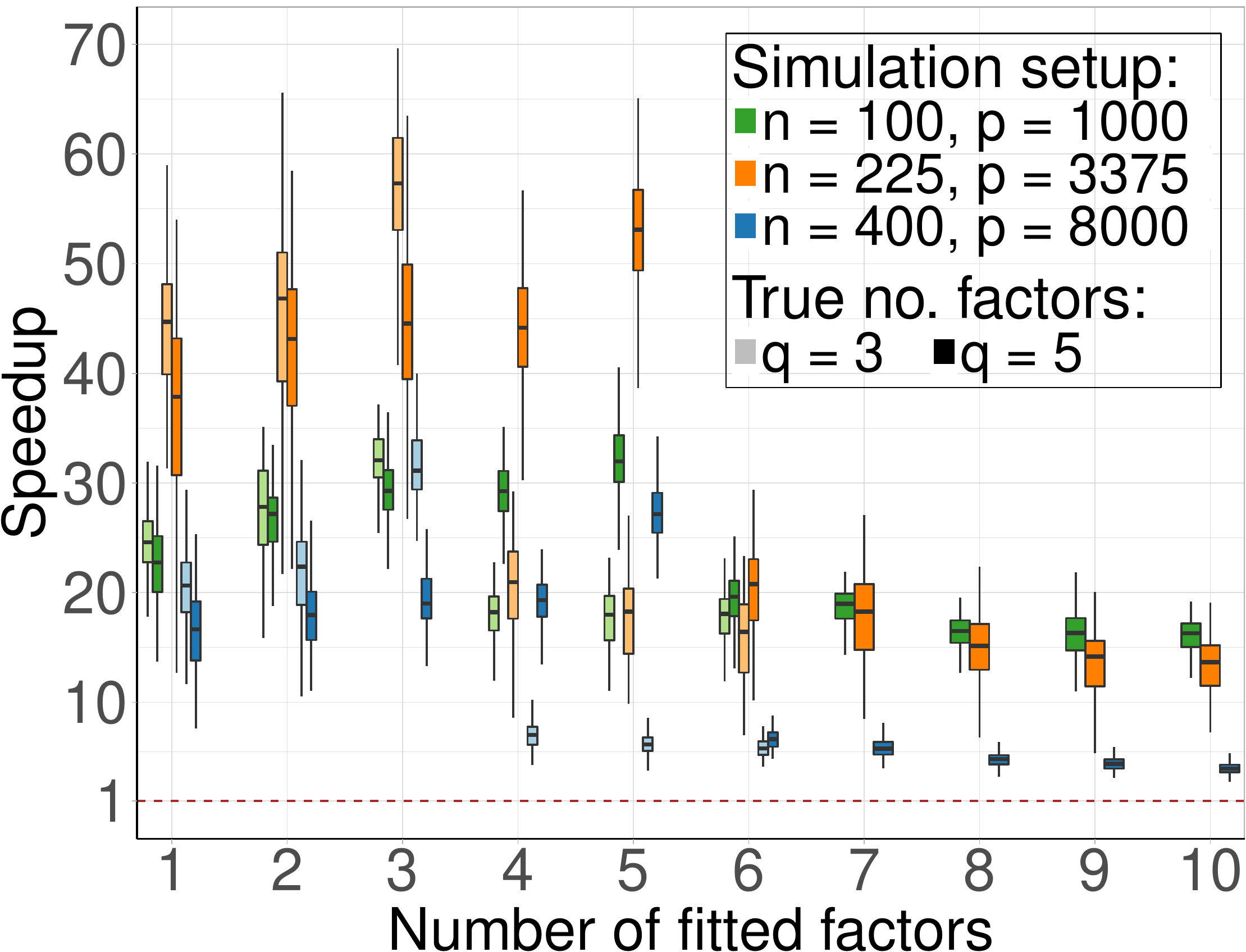}}
  \subfloat{\label{fig-time-uhd}\includegraphics[width = 0.5\textwidth]{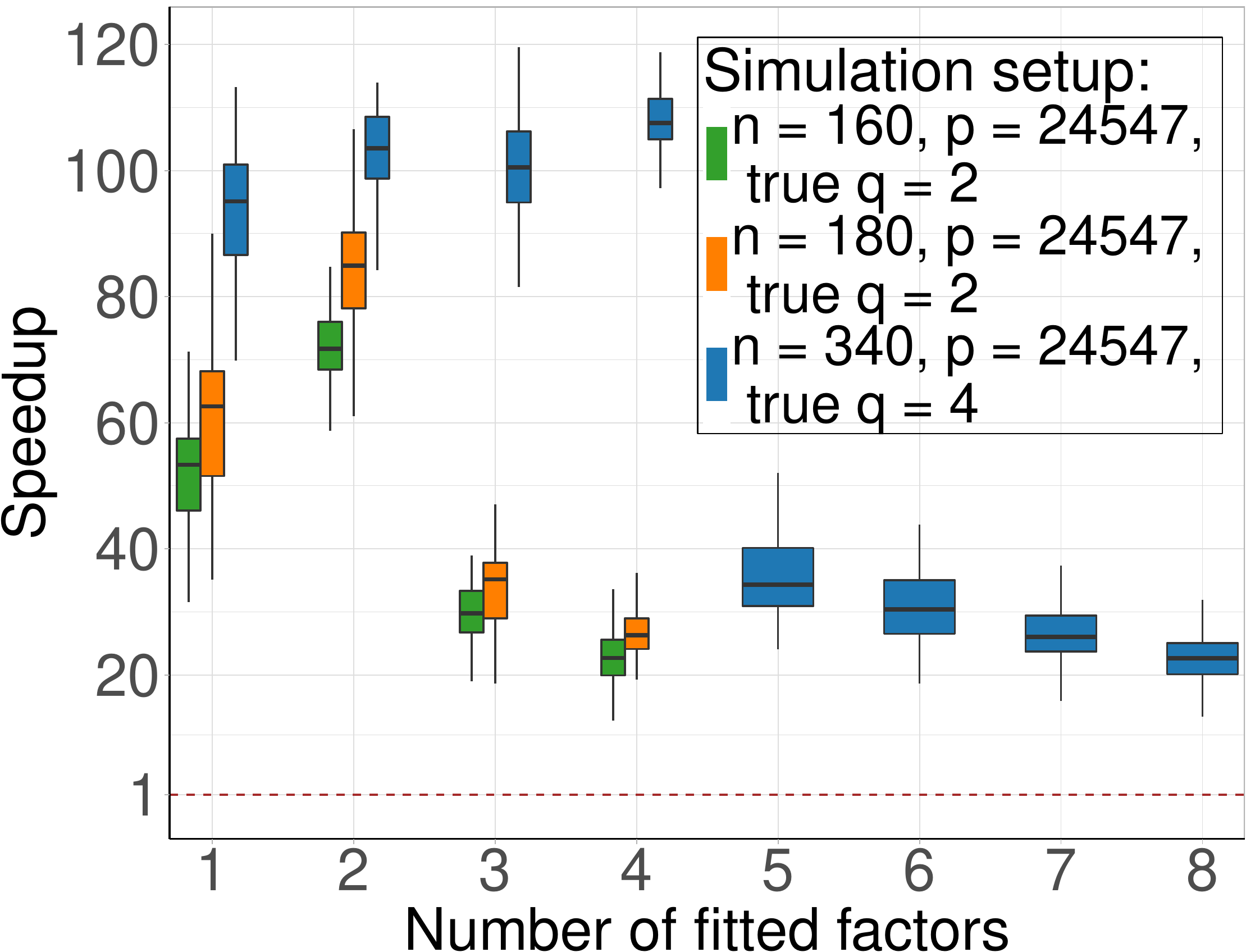}}
  }
\caption{Relative speed of FAD to EM on (left) randomly simulated and (right) data-driven cases. Lighter ones correspond to true $q=3$ and the darker ones correspond to true $q=5$.}
\label{fig-sim-time}
\vspace{-0.1in}
\end{figure}
Figure \ref{fig-sim-time} presents the relative speed of FAD to EM. Our compute times include the common initialization times. Specifically, for datasets of size $(n,p)\in\{(100,1000),(225,3375),(400,8000)\}$, FAD was 10 to 70 times faster than EM, with maximum speedup at true $q$. However, EM did not converge within 5000 iterations in any of the overfitted models. In contrast, FAD always converged but it took longer than in other cases so the speedup is underestimated because of the censoring with EM. Also, the speedup  is more pronounced (see Section \ref{sec:supp-sim-avgtime}) in the data-driven simulations where $p$ is much larger.
\subsubsection{Parameter estimation and model fit}
Under the best fitted models, FAD and EM yields identical values of $\ell_p(\hat{\Lambda},\hat{\Psi}),$ $\hat\Psi,$ $\hat\Gamma,$ and $\hat\Lambda\hat\Lambda^\top.$ Thus the relative errors (see Figure \ref{fig-sim-RG}) in estimating these parameters are also identical.


\subsection{Additional experiments in high-noise scenarios}
We conclude this section by evaluating performance in situations where
ostensibly, weak factors are hardly distinguished from high noise by
SVD methods and where EM may be preferable \citep{owen}. We applied FAD
and EM to the simulation setup of \citet{owen}: Here, the uniquenesses
were sampled from three inverse Gamma distributions with unit means
and variances of $0$, $1$ and $10$, and
$(n,p)\in\{(200,1000),(100,5000)\}$. Figure \ref{fig-sim-hn} shows
that our algorithm was substantially faster while having similar 
accuracy as EM.

\section{Suicide ideation study} \label{sec:app}
\begin{wraptable}{r}{0.7\textwidth}
\vspace{-0.25cm}
\caption{CPU times (rounded to the nearest seconds) for FAD and EM applied to the suicide ideation study dataset.} 
\label{tab:app}
\resizebox{.7\textwidth}{!}{\begin{tabular}{cc|cccccccccc}
  \toprule 
  & $\bm{q}$ & 1 & 2 & 3 & 4 & 5 & 6 & 7 & 8 & 9 & 10\\
 \midrule 
 \midrule\multirow{2}{*}{\textbf{Attempters}} & FAD & 3 & 3 & 4 & 5 & 5 & 6 & 6 & 7 & 9 & 9 \\ 
   & EM & 146 & 173 & 207 & 198 & 229 & 236 & 228 & 250 & 239 & 254 \\ 
  \midrule\multirow{2}{*}{\textbf{Ideators}} & FAD & 4 & 4 & 5 & 6 & 6 & 6 & 6 & 9 & 9 & 10 \\ 
   & EM & 118 & 197 & 207 & 200 & 222 & 244 & 241 & 226 & 258 & 258 \\ 
  \midrule\multirow{2}{*}{\textbf{Controls}} & FAD & 5 & 5 & 8 & 7 & 8 & 8 & 9 & 10 & 12 & 13 \\ 
   & EM & 300 & 451 & 456 & 407 & 426 & 461 & 483 & 438 & 566 & 519 \\ 
   \bottomrule 
\end{tabular}}
\vspace{-0.25cm}
\end{wraptable}
We applied EFA to data from \citet{just} on an fMRI study conducted while
20 words connoting negative affects were shown to 9 suicide
attempters, 8 suicide non-attempter ideators and 17 non-ideator control
subjects. For each subject-word combination, \citet{just} provided
voxel-wise per cent changes in activation relative to the baseline in
$50\!\times\!61\!\times\!23$ image volumes.  Restricting attention to
the 24547 in-brain voxels yields datasets for the attempters,
ideators and controls of sizes
$(n,p)\in\{(180,24547),(160,24547),(340,24547)\}$. We assumed  each dataset to be a random 
sample from the multinormal distribution. Our interest was in
determining if the variation in the per cent relative change in
activation for each subject type can be explained by a few 
latent factors and whether there are differences
in these factors between the three groups of subjects.

For each dataset, we performed EFA with $q = 0,1,2,\ldots,10$ factors
and using both FAD and EM. Table \ref{tab:app} demonstrates the
computational benefits of using FAD over EM. We also used BIC to
decide on the optimal $q$ ($q_{o}$) and obtained 2-factor models for both
suicide attempters and ideators, and a 4-factor model for the control
subjects. Figure~\ref{fig-fmri} provides voxel-wise displays of the
$q_o$ factor loadings, obtained using the quartimax criterion \citep{Costello}, for each type of
subject. All the factor loadings are non-negligible only in voxels
around the ventral 
 attention network (VAN) which represents one of two sensory orienting
 systems that reorient attention towards notable stimuli and is
 closely related to involuntary actions \citep{Vossel2013DorsalAV}. 
However, there are differences between the factor loadings in each
group and also among them. 
{
\captionsetup[subfigure]{font=normal,labelformat=empty}
\begin{figure}[h]
\begin{minipage}{0.92\textwidth}
\centering
\begin{mybox}[title=Attempters,nobeforeafter,width=.46\textwidth,
boxsep=0pt,left=0pt,right=0pt,top=0pt,bottom=0pt,
remember as=one]
  \subfloat[First factor]{\label{fig-fmri-atmpt1}\includegraphics[width = 0.5\textwidth]{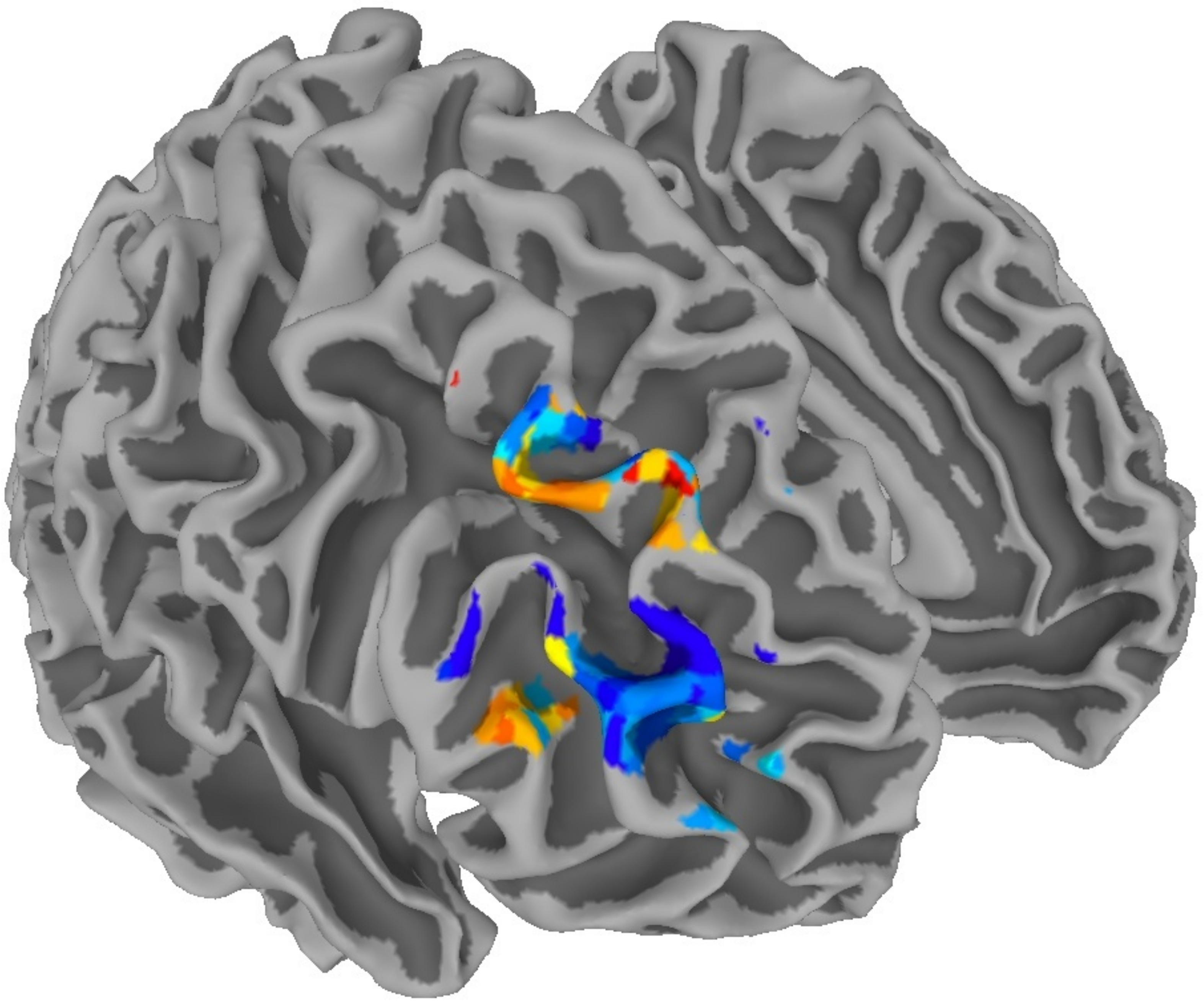}}
  \subfloat[Second factor]{\label{fig-fmri-atmpt2}\includegraphics[width = 0.5\textwidth]{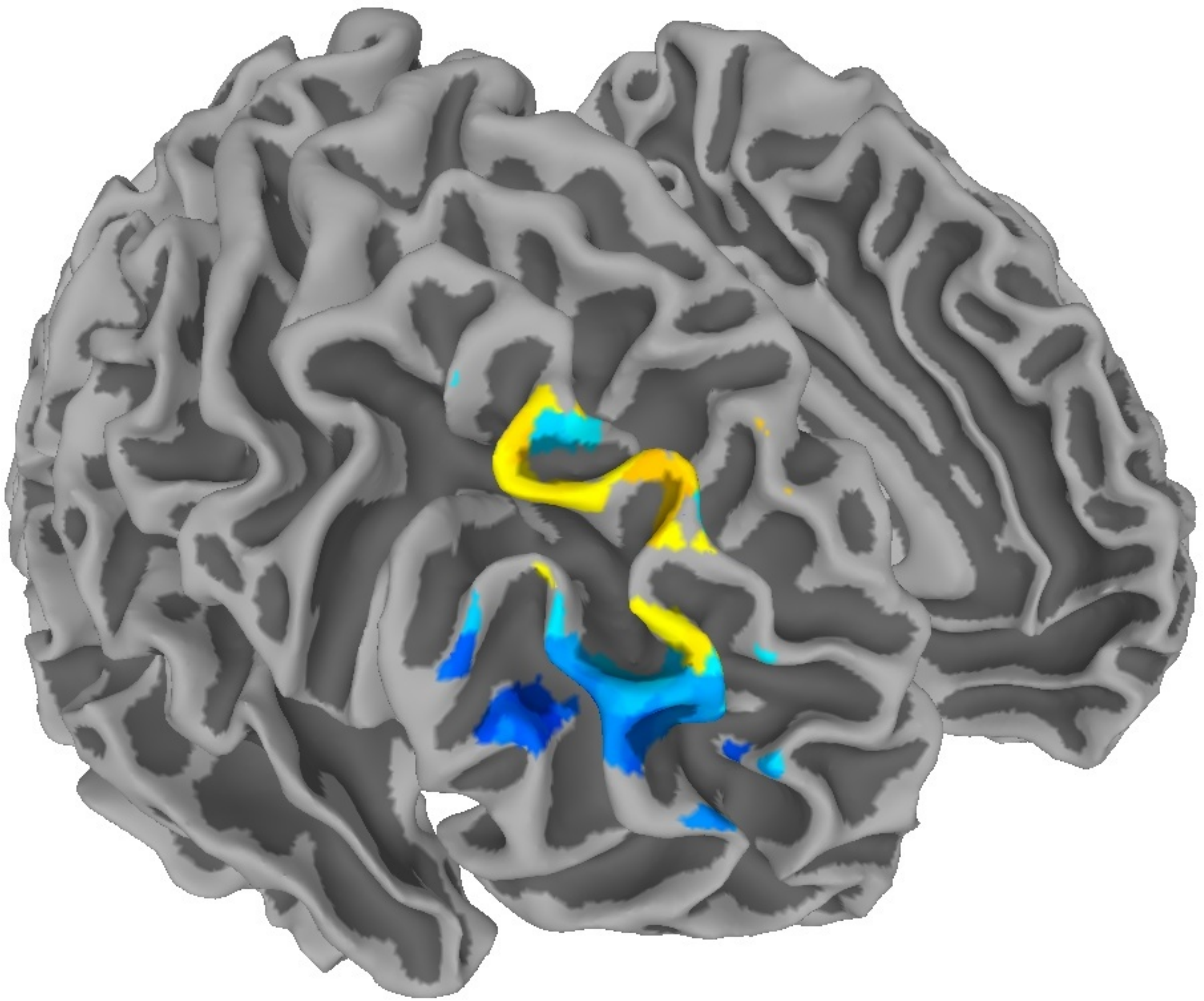}}
\end{mybox}
\begin{mybox}[title=Ideators,nobeforeafter,width=.46\textwidth,
boxsep=0pt,left=0pt,right=0pt,top=0pt,bottom=0pt,
remember as=two]
  \subfloat[First factor]{\label{fig-fmri-idt1}\includegraphics[width = 0.5\textwidth]{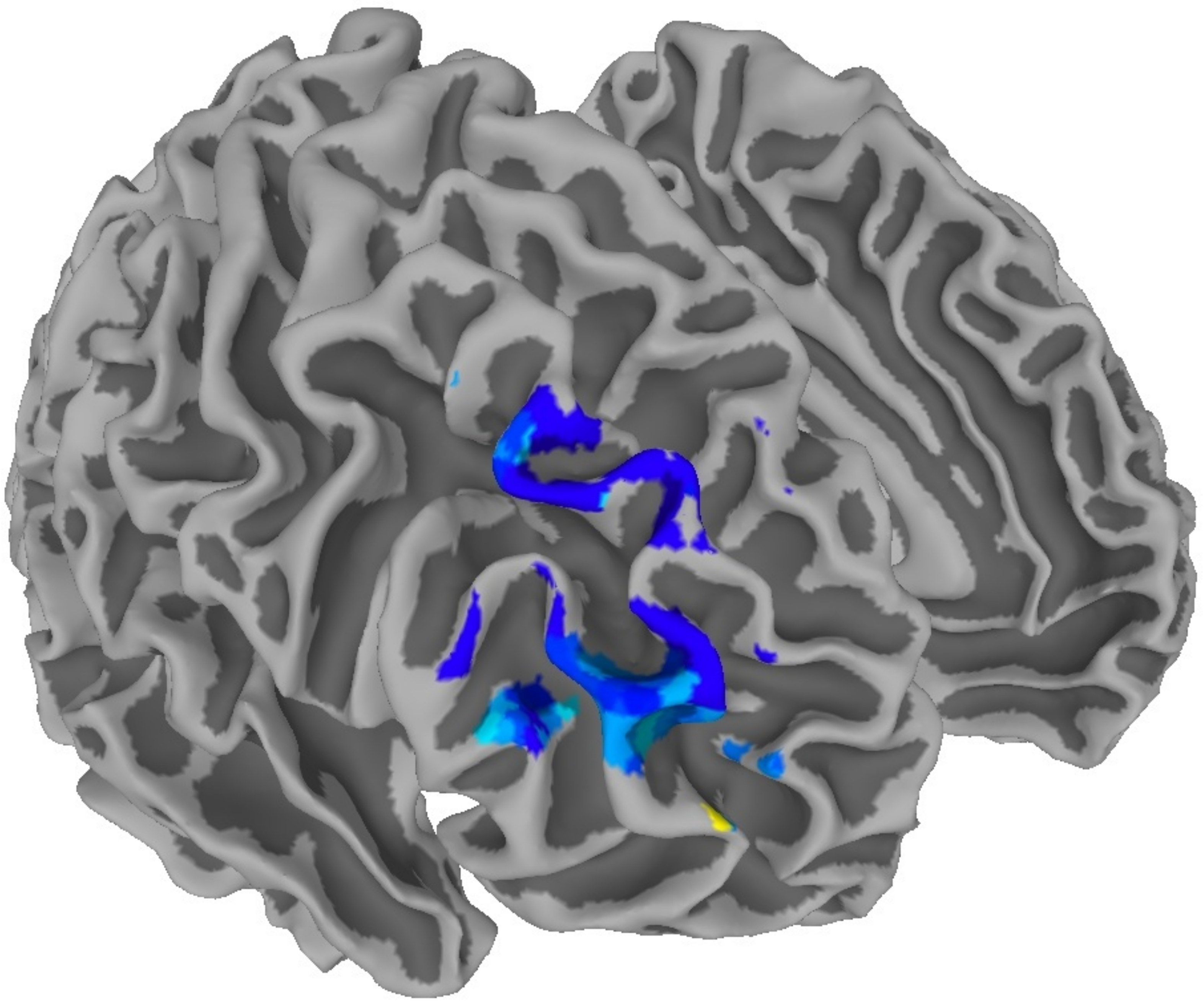}}
  \subfloat[Second factor]{\label{fig-fmri-idt2}\includegraphics[width = 0.5\textwidth]{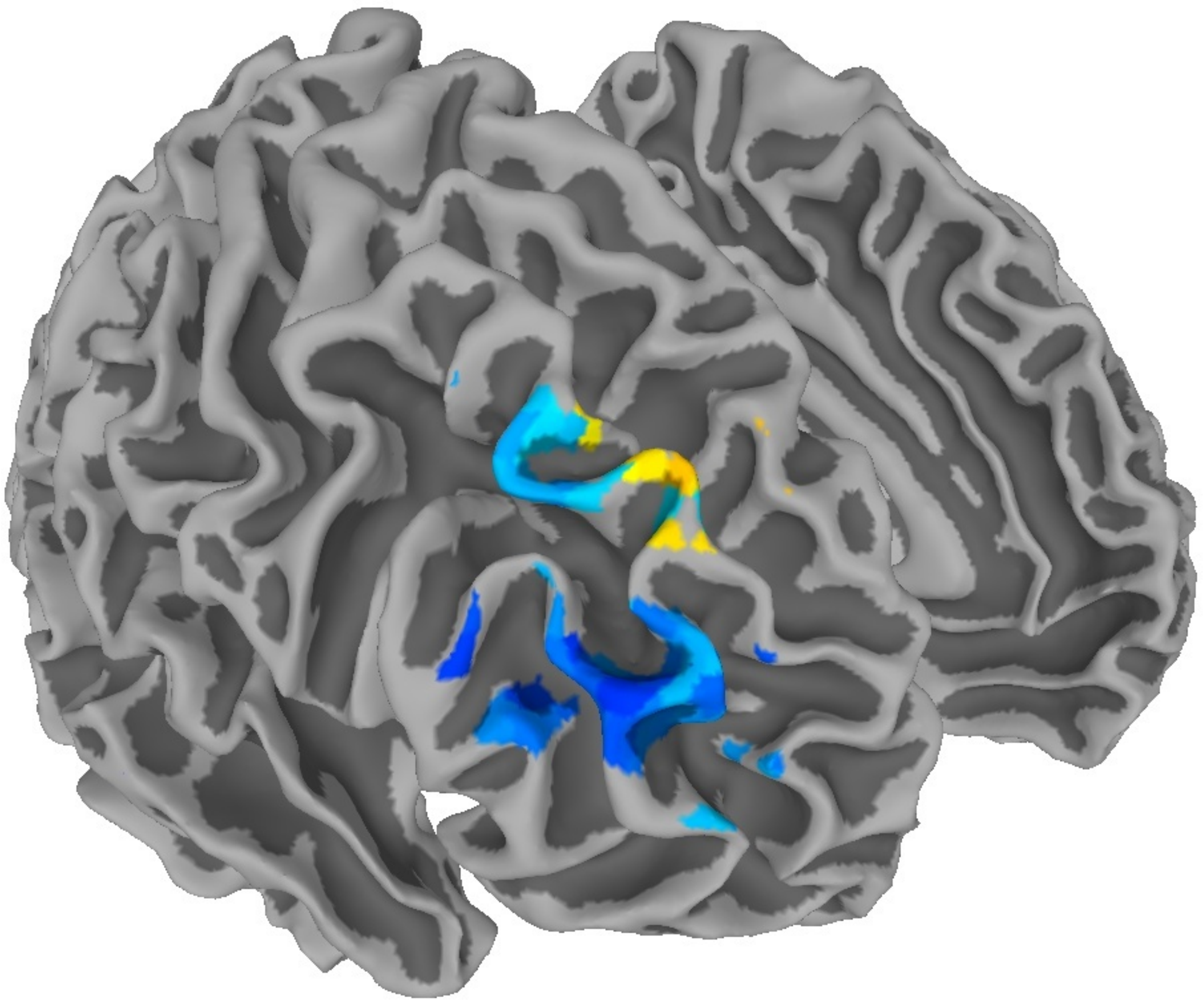}}\\
\end{mybox}
\mbox{
\begin{mybox}[title=Controls,nobeforeafter,width=.92\linewidth,
boxsep=0pt,left=0pt,right=0pt,top=0pt,bottom=0pt,
remember as=three]
  \subfloat[First factor]{\label{fig-fmri-ctrl1}\includegraphics[width = 0.25\textwidth]{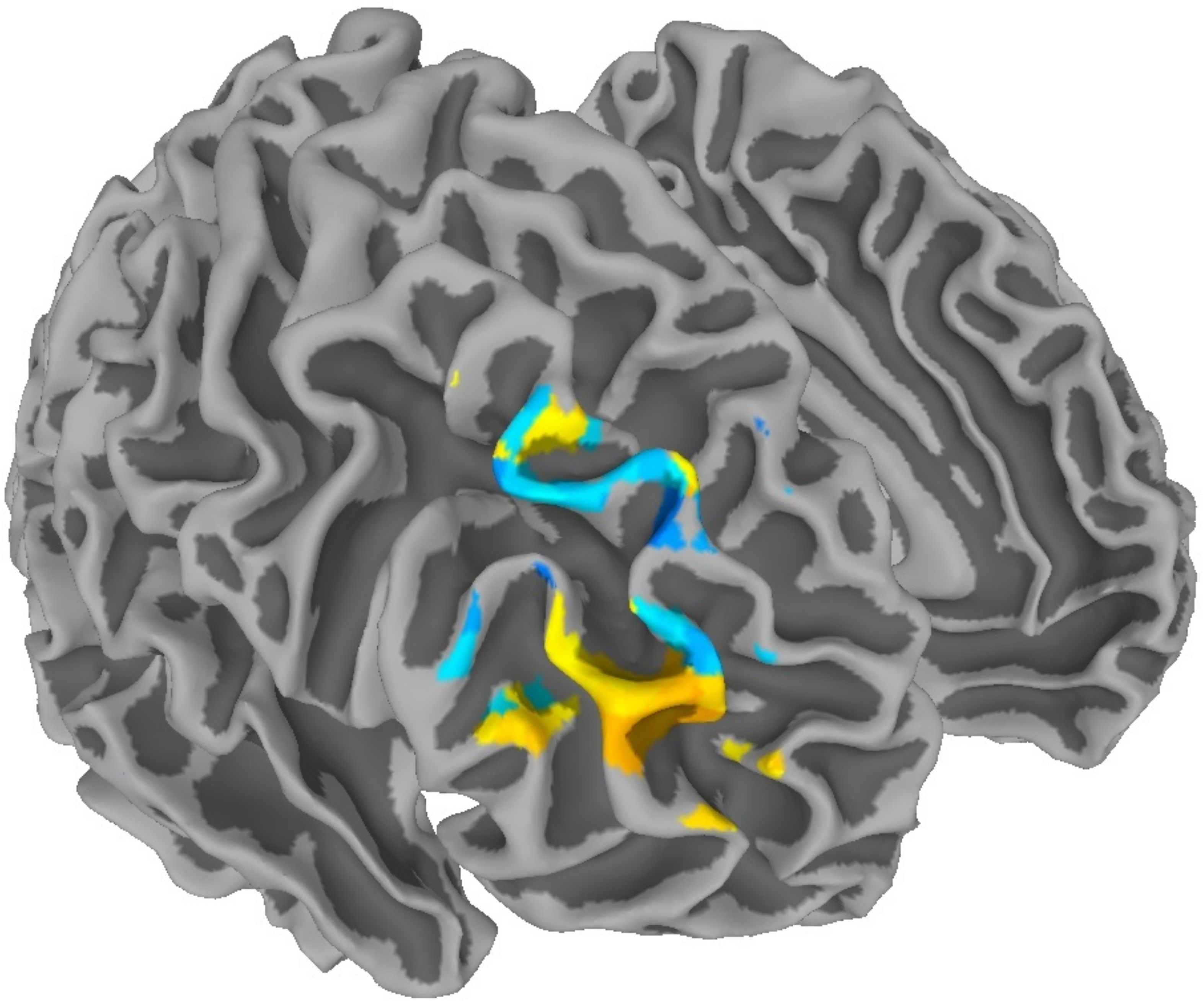}}
  \subfloat[Second factor]{\label{fig-fmri-ctrl2}\includegraphics[width = 0.25\textwidth]{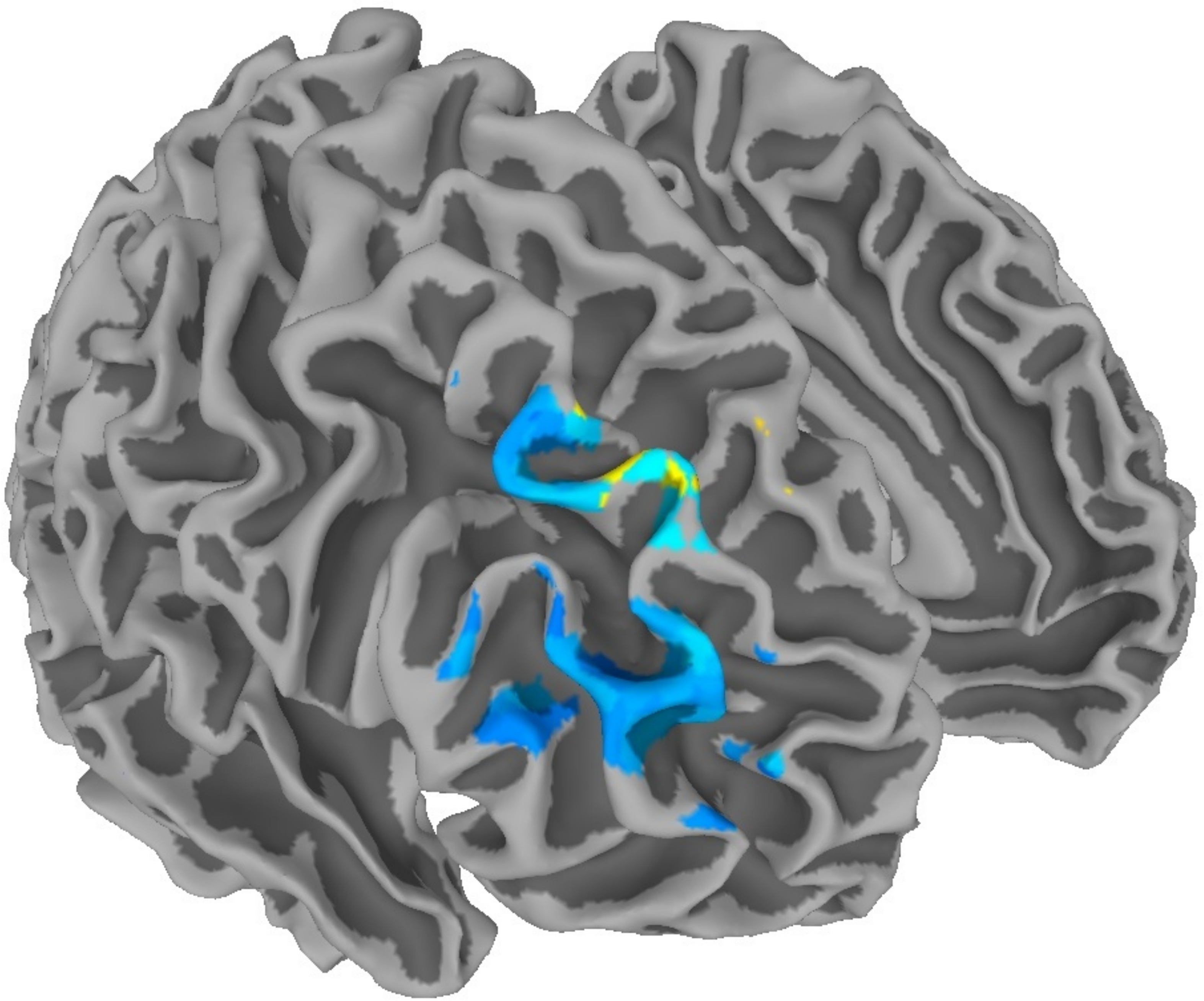}}
  \subfloat[Third factor]{\label{fig-fmri-ctrl3}\includegraphics[width = 0.25\textwidth]{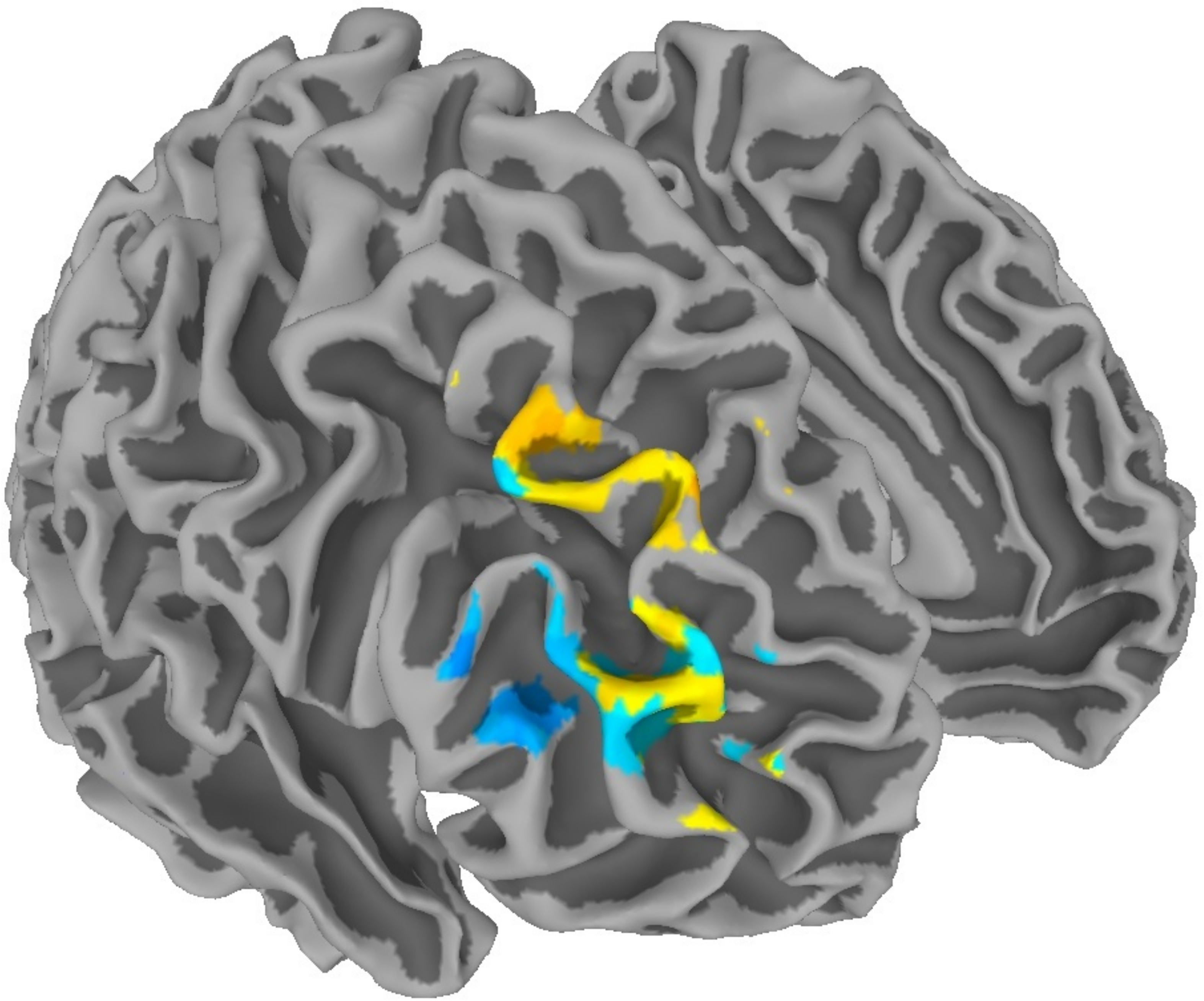}}
  \subfloat[Fourth factor]{\label{fig-fmri-ctrl4}\includegraphics[width = 0.25\textwidth]{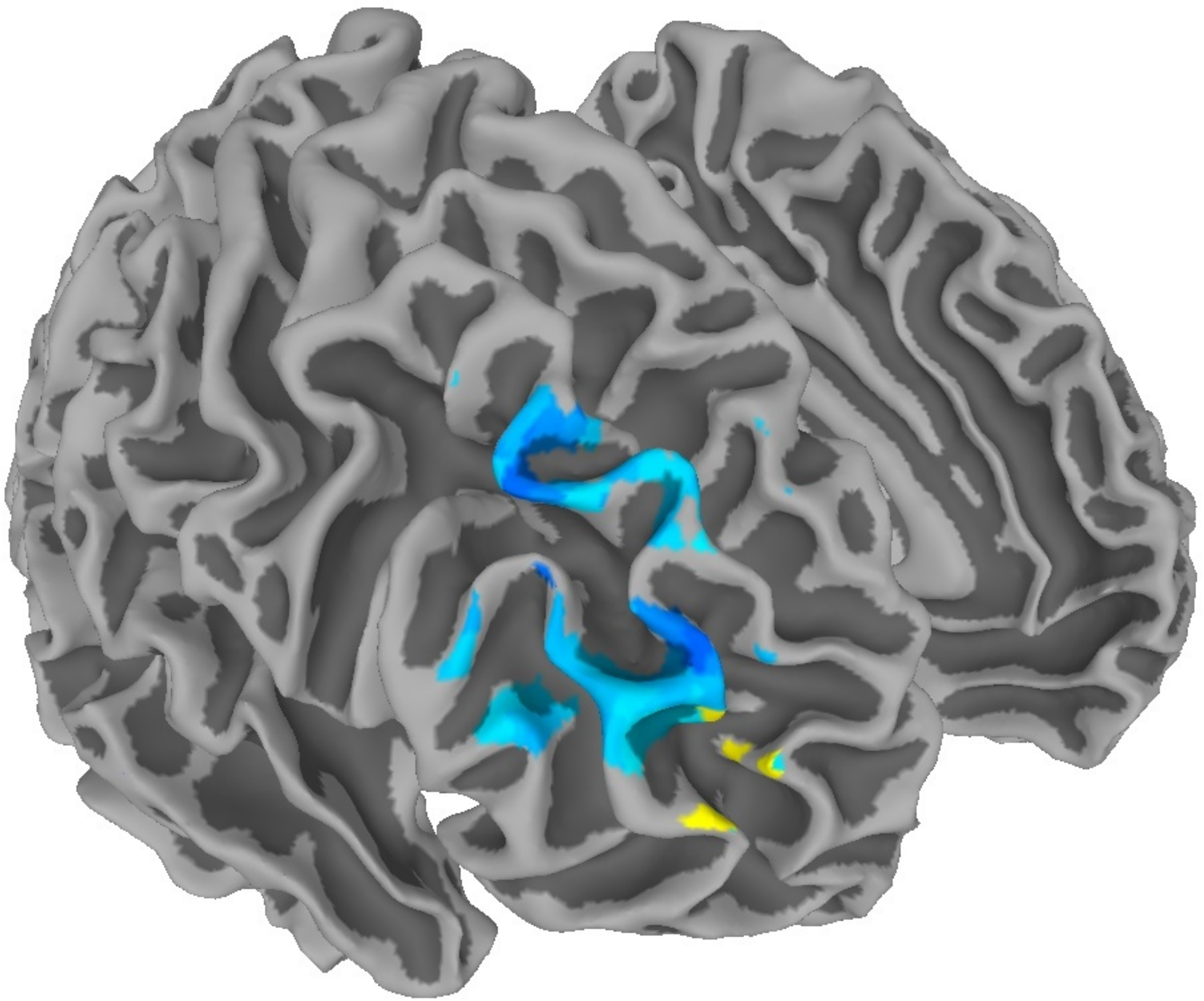}}
\end{mybox}
}
\end{minipage}\hfill
\centering
\begin{minipage}{0.08\textwidth}
\includegraphics[width=\textwidth,height=88mm]{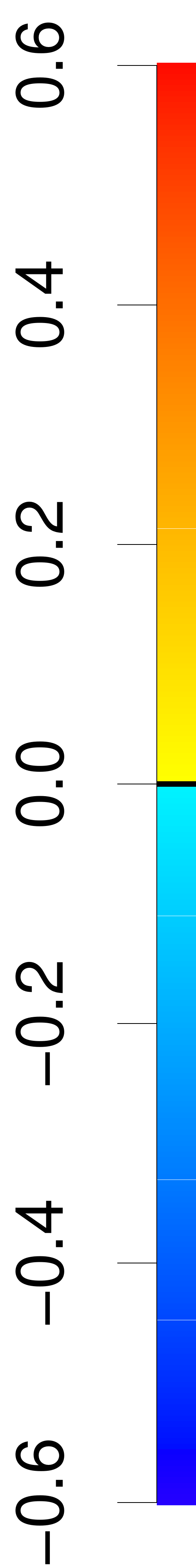}
\end{minipage}
\caption{Loading values of fitted factors for suicide attempters, ideators and controls.}
\label{fig-fmri}
\end{figure}
}
For instance, for the suicide attempters, each factor is a contrast
between different areas of the VAN, but the contrasts themselves
differ between the two factors. The first factor for the ideators is a
weighted mean of the voxels while the second factor is a contrast of
the values at the VAN voxels. For the controls, the first three
factors are different contrasts of the values at different voxels
while the fourth factor is more or less a mean of the values at these
voxels. Further,  the factor loadings in the control group are
more attenuated than for either the suicide attempters or ideators.
While a detailed analysis of our results is outside the purview
of this paper, we note that EFA has provided us with
distinct factor loadings that potentially explains the variation in suicide
attempters, non-attempter ideators and controls. However,
our analysis assumed that the image volumes are independent and Gaussian: further approaches relaxing these assumptions may be appropriate. 


\section{Discussion}
\label{sec:dis}
In this paper, we propose a new ML-based EFA method called FAD using a sophisticated computational framework that achieves both high accuracy in parameter estimation and fast convergence via matrix-free algorithms. We implement a Lanczos method for computing partial singular values and vectors and a limited-memory quasi-Newton method for ML estimation. This implementation alleviates the computational limitations of current state-of-the-art algorithms and is capable of EFA for datasets with $p>>n$. In our experiments, FAD always converged but EM struggled with overfitted models. Although not demonstrated in this paper, FAD is also well-suited for distributed computing systems because it only uses the data matrix for computing matrix-vector products. FAD paves the way to develop fast methods for mixtures of factor analyzers and factor models for non-Gaussian data in high-dimensional clustering and classification problems.

\ifCLASSOPTIONcaptionsoff
\fi
\setcounter{equation}{0}
\setcounter{figure}{0}
\setcounter{table}{0}
\setcounter{section}{0}
\renewcommand{\thesection}{S\arabic{section}}
\renewcommand{\thesubsection}{\thesection.\arabic{subsection}}
\renewcommand{\theequation}{S\arabic{equation}}
\renewcommand{\thefigure}{S\arabic{figure}}
\renewcommand{\thetable}{S\arabic{table}}
\section{Supplementary materials for Methodology}
\label{sec:sim-meth}
\subsection{The EM algorithm for factor analysis on Gaussian data}
\label{sec:supp-em}
The complete data log-likelihood function is
\begin{equation} \label{eq:1}
\begin{split}
        \ell_{C}(\bm{\mu},\bm{\Lambda},\bm{\Psi}) 
        = c & -\frac{n}{2}\log\det\Psi
         - \frac{1}{2}\sum^{n}_{i=1}\{ (\mathbf{Y}_i- \bm{\mu} - \bm{\Lambda}\mathbf{Z}_i)^{\top}\Psi^{-1}(\mathbf{Y}_i- \bm{\mu} - \bm{\Lambda}\mathbf{Z}_i)\}\\
        = c & -\frac{n}{2}\log\det\Psi
         - \frac{1}{2}\Tr\{\bm{\Psi}^{-1}\sum^{n}_{i=1}{(\mathbf{Y}_i- \bm{\mu})(\mathbf{Y}_i- \bm{\mu})^{\top}} - 2\bm{\Psi}^{-1}\bm{\Lambda}\sum^{n}_{i=1}{\mathbf{Z}_i}(\mathbf{Y}_i- \bm{\mu})^{\top}\}\\
         & - \frac12\Tr\{\bm{\Lambda}^{\top}\bm{\Psi}^{-1}\bm{\Lambda}\sum^{n}_{i=1}\mathbf{Z}_i\mathbf{Z}_i^{\top}\},
\end{split}
\end{equation}
where $c$ is a constant that does not depend on the parameters.

\subsubsection{E-Step computations}
Since the ML estimate of $\bm{\mu}$ is $\bar{\mathbf{Y}}$, at the current estimates $\bm{\Lambda}_t$ and $\bm{\Psi}_t$ , the expected complete log-likelihood or so called $\mathrm{Q}$ function is given by
\begin{equation} \label{eq:2}
\begin{split}
        \mathrm{Q}(\bm{\Lambda}_{t+1},\bm{\Psi}_{t+1}|\bar{\mathbf{Y}},\bm{\Lambda}_t,\bm{\Psi}_t) 
         = &\mathrm{E}[\ell_{C}(\bm{\Lambda},\bm{\Psi}|\mathbf{Y},\bm{\Lambda}_t,\bm{\Psi}_t]\\
        = &-\frac{n}{2}\log\det\Psi
         - \frac{n}{2}\Tr\bm{\Psi}^{-1}\mathbf{S} - \Tr\{\bm{\Psi}^{-1}\bm{\Lambda}\sum^{n}_{i=1}{\mathrm{E}[\mathbf{Z}_i|\mathbf{Y},\bm{\Lambda}_t,\bm{\Psi}_t]}(\mathbf{Y}_i- \bar{\mathbf{Y}})^{\top}\}\\
         &+ \frac12\Tr\{\bm{\Lambda}^{\top}\bm{\Psi}^{-1}\bm{\Lambda}\frac{1}{n}\sum^{n}_{i=1}\mathrm{E}[\mathbf{Z}_i\mathbf{Z}_i^{\top}|\mathbf{Y},\bm{\Lambda}_t,\bm{\Psi}_t]\}.
\end{split}
\end{equation}
Since $\mathbf{Z}_i|\mathbf{Y},\bm{\Lambda}_t,\bm{\Psi}_t \sim \mathcal{N}_{q}(\bm{\Lambda}_t^{\top}\mathbf{\Sigma}_t^{-1}(\mathbf{Y}_i-\bar{\mathbf{Y}}),(\mathbf{I}_q+\bm{\Lambda}_t^{\top}\bm{\Psi}_t^{-1}\bm{\Lambda}_t)^{-1})$. Then,
\begin{equation} \label{eq:3}
    \begin{split}
    &\mathrm{E}[\mathbf{Z}_i|\mathbf{Y},\bm{\Lambda}_t,\bm{\Psi}_t]
    = \bm{\Lambda}_{t}^{\top}\mathbf{\Sigma}_{t}^{-1}(\mathbf{Y}_i-\bar{\mathbf{Y}})\\
    \end{split}
\end{equation}
and
\begin{equation} \label{eq:4}
    \begin{split}
    \mathrm{E}[\mathbf{Z}_i\mathbf{Z}_i^{\top}|\mathbf{Y},\bm{\Lambda}_t,\bm{\Psi}_t]
    = & \mathrm{Var}[\mathbf{Z}_i|\mathbf{Y},\bm{\Lambda}_t,\bm{\Psi}_t] + \mathrm{E}[\mathbf{Z}_i|\mathbf{Y},\bm{\Lambda}_t,\bm{\Psi}_t]\mathrm{E}[\mathbf{Z}_i^{\top}|\mathbf{Y}_i,\bar{\mathbf{Y}},\bm{\Lambda}_t,\bm{\Psi}_t]\\
    = & (\mathbf{I}_q+\bm{\Lambda}_t^{\top}\bm{\Psi}_t^{-1}\bm{\Lambda}_t)^{-1} + \bm{\Lambda}_t^{\top}\mathbf{\Sigma}_t^{-1}(\mathbf{Y}_i- \bar{\mathbf{Y}})(\mathbf{Y}_i- \bar{\mathbf{Y}})^{\top}\mathbf{\Sigma}_t^{-1}\bm{\Lambda}_t.
\end{split}
\end{equation}

\subsubsection{M-Step computations}
The parameters $\bm{\Lambda}_{t+1}$ and $\bm{\Psi}_{t+1}$ are obtained by maximizing  $\mathrm{Q}(\bm{\Lambda}_{t+1},\bm{\Psi}_{t+1}|\bar{\mathbf{Y}},\bm{\Lambda}_t,\bm{\Psi}_t)$ following equation \ref{eq:2}. Specifically, given ${\mathbf{Y}}$, $\bm{\Lambda}_t$ and $\bm{\Psi}_t$, the maximizer $\bm{\Lambda}_{t+1}$ is given by
\begin{equation} \label{eq:6}
\begin{split}
\hat{\bm{\Lambda}}_{t+1} 
= & 
(\frac{1}{n}\sum^{n}_{i=1}(\mathbf{Y}_i- \bar{\mathbf{Y}}){\mathrm{E}[\mathbf{Z}_i^{\top}|\mathbf{Y},\bm{\Lambda}_t,\bm{\Psi}_t]})(\frac{1}{n}\sum^{n}_{i=1}{\mathrm{E}[\mathbf{Z}_i\mathbf{Z}_i^{\top}|\mathbf{Y},\bm{\Lambda}_t,\bm{\Psi}_t]})^{-1}\\
 = & \mathbf{S}\mathbf{\Sigma}^{-1}_t\bm{\Lambda}_t((\mathbf{I}_q+\bm{\Lambda}^{\top}_t\bm{\Psi}^{-1}_t\bm{\Lambda}_t)^{-1} + \bm{\Lambda}^{\top}_t\mathbf{\Sigma}^{-1}_t\mathbf{S}\mathbf{\Sigma}^{-1}_t\bm{\Lambda}_t)^{-1}\\
\end{split}
\end{equation}
where $\mathbf{\Sigma}_t = \bm{\Lambda}_t\bm{\Lambda}^{\top}_t + \bm{\Psi}_t$. By Woodbury matrix identity \citep{hendersonandsearle81}, $\mathbf{\Sigma}^{-1}_t\bm{\Lambda}_t = \bm{\Psi}^{-1}_t\bm{\Lambda}_t(\mathbf{I}_q+\bm{\Lambda}^{\top}_t\bm{\Psi}^{-1}_t\bm{\Lambda}_t)^{-1}$, so \ref{eq:6} can be simplified as
\begin{equation} \label{eq:7}
\begin{split}
\hat{\bm{\Lambda}}_{t+1} 
 = & \mathbf{S}\bm{\Psi}^{-1}_t\bm{\Lambda}_t(\mathbf{I}_q+\bm{\Lambda}^{\top}_t\bm{\Psi}^{-1}_t\bm{\Lambda}_t)^{-1}((\mathbf{I}_q+\bm{\Lambda}^{\top}_t\bm{\Psi}^{-1}_t\bm{\Lambda}_t)^{-1} + \bm{\Lambda}^{\top}_t\mathbf{\Sigma}^{-1}_t\mathbf{S}\bm{\Psi}^{-1}_t\bm{\Lambda}_t(\mathbf{I}_q+\bm{\Lambda}^{\top}_t\bm{\Psi}^{-1}_t\bm{\Lambda}_t)^{-1})^{-1}\\
 = & \mathbf{S}\bm{\Psi}^{-1}_t\bm{\Lambda}_t(\mathbf{I}_q + \bm{\Lambda}^{\top}_t\mathbf{\Sigma}^{-1}_t\mathbf{S}\bm{\Psi}^{-1}_t\bm{\Lambda}_t)^{-1}.
\end{split}
\end{equation}
Next, given $\mathbf{Y},$ $\bm{\Lambda}_t$, $\bm{\Psi}_t$ and $\hat{\bm{\Lambda}}_{t+1}$, the ML estimate of $\bm{\Psi}_{t+1}$ is given by
\begin{equation} 
\begin{split}
\hat{\bm{\Psi}}_{t+1} 
= 
\mathrm{diag}\Big(&\mathbf{S} - \frac{2}{n}\sum^{n}_{i=1}(\mathbf{Y}_i- \bar{\mathbf{Y}}){\mathrm{E}[\mathbf{Z}_i^{\top}|\mathbf{Y},\bm{\Lambda}_t,\bm{\Psi}_t]}\hat{\bm{\Lambda}}^{\top}_{t+1} \\
& + \hat{\bm{\Lambda}}_{t+1}\frac{1}{n}\sum^{n}_{i=1}{\mathrm{E}[\mathbf{Z}_i\mathbf{Z}_i^{\top}|\mathbf{Y},\bm{\Lambda}_t,\bm{\Psi}_t]})^{-1}\hat{\bm{\Lambda}}^{\top}_{t+1}\Big).
\end{split}
\end{equation}
Substitute with \ref{eq:6}, we get
\begin{equation} 
\begin{split}
\hat{\bm{\Psi}}_{t+1} 
= &
\mathrm{diag}\Big(\mathbf{S} - \frac{2}{n}\sum^{n}_{i=1}(\mathbf{Y}_i- \bar{\mathbf{Y}}){\mathrm{E}[\mathbf{Z}_i^{\top}|\mathbf{Y},\bm{\Lambda}_t,\bm{\Psi}_t]}\hat{\bm{\Lambda}}^{\top}_{t+1} \\
& \qquad\ + \frac{1}{n}\sum^{n}_{i=1}(\mathbf{Y}_i- \bar{\mathbf{Y}}){\mathrm{E}[\mathbf{Z}_i^{\top}|\mathbf{Y},\bm{\Lambda}_t,\bm{\Psi}_t]}\hat{\bm{\Lambda}}^{\top}_{t+1}\Big)\\
= & \mathrm{diag}\Big(\mathbf{S} - 2\mathbf{S}\mathbf{\Sigma}^{-1}_t\bm{\Lambda}_t\hat{\bm{\Lambda}}^{\top}_{t+1}\Big).
\end{split}
\end{equation}

\subsection{Proof of Lemma \ref{lemma:profileout}}
\label{sec:supp-lemma}
From equation \ref{eqn:loglikelihood}, the ML estimates of $\bm{\Lambda}$ and $\bm{\Psi}$ are obtained by solving the score equations
\begin{equation} \label{eq:system_eq}
\begin{split}
\begin{cases}
    & \bm{\Lambda}(\mathbf{I}_q + \bm{\Lambda}^\top\bm{\Psi}^{-1}\bm{\Lambda}) = 
    \mathbf{S}\bm{\Psi}^{-1}\bm{\Lambda}\\
    & \bm{\Psi} = \mathrm{diag}(\mathbf{S}-\bm{\Lambda}\bm{\Lambda}^\top)
    \end{cases}
    \end{split}
\end{equation}
From $\bm{\Lambda}(\mathbf{I}_q + \bm{\Lambda}^\top\bm{\Psi}^{-1}\bm{\Lambda}) = 
    \mathbf{S}\bm{\Psi}^{-1}\bm{\Lambda}$, we have
    \begin{equation}\label{eq:profileout_Lambda}
       \bm{\Psi}^{-1/2}\bm{\Lambda}(\mathbf{I}_q + (\bm{\Psi}^{-1/2}\bm{\Lambda})^\top\bm{\Psi}^{-1/2}\bm{\Lambda}) = \bm{\Psi}^{-1/2}\mathbf{S}\bm{\Psi}^{-1/2}\bm{\Psi}^{-1/2}\bm{\Lambda}.
    \end{equation}
Suppose that $\bm{\Psi}^{-1/2}\mathbf{S}\bm{\Psi}^{-1/2} = \mathbf{V}\mathbf{D}\mathbf{V}^\top$ and that the diagonal elements in $\mathbf{D}$ are in decreasing order with $\theta_1\geq\theta_2\geq,\cdots,\geq\theta_p$. Let $\mathbf{D} = \begin{bmatrix} 
\mathbf{D}_q & 0 \\
0 & \mathbf{D}_m 
\end{bmatrix}$ with $m=p-q$ and $\mathbf{D}_q$ containing the largest $q$ eigenvalues  $\theta_1\geq\theta_2\geq,\cdots,\geq\theta_q$. The corresponding $q$ eigenvectors forms columns of matrix $\mathbf{V}_q$ so that $\mathbf{V} = [\mathbf{V}_q,\mathbf{V}_m]$. Then, if $\mathbf{D}_q > \mathbf{I}_q$, \ref{eq:profileout_Lambda} shows that
\begin{equation}
     \bm{\Lambda} =  \bm{\Psi}^{1/2}\mathbf{V}_q(\mathbf{D}_q-\mathbf{I}_q)^{1/2}.
\end{equation}
The square roots of $\theta_1,\cdots,\theta_q$ are the $q$ largest singular values of $n^{-1/2}(\mathbf{Y}-\mathbf{1}\bar{\mathbf{Y}}^\top)\bm{\Psi}^{-1/2}$ and columns in $\mathbf{V}_q$ are then the corresponding $q$ right-singular vectors. Hence, conditional on $\bm{\Psi}$, $\bm{\Lambda}$ is maximized at $\hat{\bm{\Lambda}} = \bm{\Psi}^{1/2}\mathbf{V}_q\bm{\Delta}$, where $\bm{\Delta}$ is a diagonal matrix with elements $\mathrm{max}(\theta_i-1,0)^{1/2}, i = 1,\cdots,q$.

From the construction of $\mathbf{V}_q$ and $\mathbf{V}_m$, we have $\mathbf{V}^\top_q\mathbf{V}_q = \mathbf{I}_q\text{, }
\mathbf{V}^\top_m\mathbf{V}_m = \mathbf{I}_m\text{, }
\mathbf{V}_q\mathbf{V}^\top_q + \mathbf{V}_m\mathbf{V}^\top_m = \mathbf{I}_p\text{, }\mathbf{V}^\top_q\mathbf{V}_m = \boldsymbol{0}$ and hence, $(\mathbf{V}_q\mathbf{D}_q\mathbf{V}^\top_q +  \mathbf{V}_m\mathbf{V}^\top_m)(\mathbf{V}_q\mathbf{D}^{-1}_q\mathbf{V}^\top_q +  \mathbf{V}_m\mathbf{V}^\top_m) = \mathbf{I}_p$.

Let $\mathbf{A} = \mathbf{V}_q\bm{\Delta}^2\mathbf{V}^\top_q$. Then $\mathbf{A}\mathbf{A} = \mathbf{V}_q\bm{\Delta}^4\mathbf{V}^\top_q$ and
\begin{equation}\label{eq:matrix_results1}
    \begin{split}
       |\mathbf{A}+\mathbf{I}_p| 
       = & |(\mathbf{A}+\mathbf{I}_p)\mathbf{A}|/|\mathbf{A}| \\
       = & |\mathbf{V}_q(\bm{\Delta}^4+\bm{\Delta}^2)\mathbf{V}^\top_q|/|\mathbf{V}_q\bm{\Delta}^2\mathbf{V}^\top_q| \\
       = & |\bm{\Delta}^2+\mathbf{I}_q| = \prod_{j=1}^{q}{\theta_j}
       \end{split}
\end{equation}
and
\begin{equation}\label{eq:matrix_results2}
      \begin{split}
      (\mathbf{A}+\mathbf{I}_p)^{-1} 
      = & (\mathbf{V}_q\bm{\Delta}^2\mathbf{V}^\top_q + \mathbf{V}_q\mathbf{V}^\top_q + \mathbf{V}_m\mathbf{V}^\top_m)^{-1}\\
      = & (\mathbf{V}_q(\bm{\Delta}^2+\mathbf{I}_q)\mathbf{V}^\top_q +  \mathbf{V}_m\mathbf{V}^\top_m)^{-1}\\
      = & (\mathbf{V}_q\mathbf{D}_q\mathbf{V}^\top_q +  \mathbf{V}_m\mathbf{V}^\top_m)^{-1}\\
      = & \mathbf{V}_q\mathbf{D}^{-1}_q\mathbf{V}^\top_q +  \mathbf{V}_m\mathbf{V}^\top_{m}.
    \end{split}
\end{equation}
Based on \ref{eq:matrix_results1} and \ref{eq:matrix_results2} and equation \ref{eqn:profilelikelihood}, the profile log-likelihood is given by
\begin{equation}\label{eq:proof_lemma1}
    \begin{split}
        \ell_{p}(\bm{\Psi}) 
        = c & - \frac{n}{2}\log{|\hat{\bm{\Lambda}}\hat{\bm{\Lambda}}^\top + \bm{\Psi}|} -\frac{n}{2}\Tr(\hat{\bm{\Lambda}}\hat{\bm{\Lambda}}^\top + \bm{\Psi})^{-1}\mathbf{S}\\
                = c & - \frac{n}{2}\Big\{\log|\bm{\Psi}^{1/2}(\mathbf{V}_q\bm{\Delta}^2\mathbf{V}^\top_q + \mathbf{I}_p)\bm{\Psi}^{1/2}| + \Tr\bm{\Psi}^{1/2}(\mathbf{V}_q\bm{\Delta}^2\mathbf{V}^\top_q + \mathbf{I}_p)\bm{\Psi}^{1/2})^{-1}\mathbf{S}\Big\}\\
        = c & - \frac{n}{2}\Big\{ \log\det\Psi + \log{|\mathbf{V}_q\bm{\Delta}^2\mathbf{V}^\top_q + \mathbf{I}_p|} + \Tr(\mathbf{V}_q\mathbf{D}^{-1}_q\mathbf{V}^\top_q +  \mathbf{V}_m\mathbf{V}^\top_m)\bm{\Psi}^{-1/2}\mathbf{S}\bm{\Psi}^{-1/2}
        \Big\}\\
        = c & - \frac{n}{2}\Big\{  \log\det\Psi + \sum_{j=1}^{q}{\log{\theta_j}} + \Tr(\mathbf{D}^{-1}_q\mathbf{V}^\top_q\mathbf{V}\mathbf{D}\mathbf{V}^\top\mathbf{V}_q + \Tr\mathbf{V}^\top_m\mathbf{V}\mathbf{D}\mathbf{V}^\top\mathbf{V}_m\Big\}\\
        = c & - \frac{n}{2}\Big\{  \log\det\Psi + \sum_{j=1}^{q}{\log{\theta_j}} + \Tr\mathbf{D}^{-1}_q\mathbf{D}_q + \Tr\mathbf{D}_m\Big\}\\
        = c &- \frac{n}{2}\Big\{ \log\det\Psi + \sum_{j=1}^{q}{\log{\theta_j}} + q + \Tr\bm{\Psi}^{-1}\mathbf{S}- \sum_{j=1}^{q}{\theta_j}\Big\}.
    \end{split}
\end{equation}

\section{Additional results for simulation studies}
\label{sec:supp-sim}
\subsection{Average CPU time}
\label{sec:supp-sim-avgtime}
\begin{table}[ht]
\caption{Average CPU time (in seconds) of FAD and EM applied with 1-6 factors for randomly simulated datasets where true $q=3$.} 
\label{table:app}
\centering
\resizebox{.7\textwidth}{!}{\begin{tabular}{cc|cccccc}
  \toprule 
  & & 1 & 2 & 3 & 4 & 5 & 6 \\
 \midrule 
 \midrule\multirow{2}{*}{\bm{$(n,p,q)=(10^2,10^3,3)$}} & FAD & 0.101 &0.092 &0.096 &0.116 &0.122 &0.128\\ 
   & EM & 2.494 &2.519 &3.076 &2.012 &2.075 &2.162\\ 
  \midrule\multirow{2}{*}{\bm{$(n,p,q)=(15^2,15^3,3)$}} & FAD & 0.639 &0.514 &0.486 &0.841 &0.966 &1.025\\
   & EM & 24.798 &22.885 &27.906 &16.822 &16.630 &15.722\\ 
  \midrule\multirow{2}{*}{\bm{$(n,p,q)=(20^2,20^3,3)$}} & FAD & 2.933 &2.658 &2.580 &7.135 &7.863 &8.590\\ 
   & EM & 57.052 &57.527 &81.463 &49.689 &48.508 &49.504\\ 
   \bottomrule 
\end{tabular}}
\end{table}

\begin{table}[ht]
\caption{Average CPU time (in seconds) of FAD and EM applied with 1-10 factors for randomly simulated datasets where true $q=5$.} 
\label{table:app2}
\centering
\resizebox{.95\textwidth}{!}{\begin{tabular}{cc|cccccccccc}
  \toprule 
  & & 1 & 2 & 3 & 4 & 5 & 6 & 7 & 8 & 9 & 10\\
 \midrule 
 \midrule\multirow{2}{*}{\bm{$(n,p,q)=(10^2,10^3,5)$}}
   & FAD &0.102 &0.095 &0.096 &0.096 &0.094 &0.119 &0.124 &0.128 &0.134 &0.137\\ 
   & EM &2.290 &2.539 &2.808 &2.800 &2.985 &2.301 &2.327 &2.097 &2.166 &2.196\\ 
  \midrule\multirow{2}{*}{\bm{$(n,p,q)=(15^2,15^3,5)$}} & FAD & 0.667 &0.513 &0.501 &0.507 &0.497 &0.828 &0.919 &1.039 &1.108 &1.143\\ 
   & EM &22.545 &21.066 &22.300 &22.197 &26.300 &16.544 &15.796 &14.767 &14.292 &14.789\\ 
  \midrule\multirow{2}{*}{\bm{$(n,p,q)=(20^2,20^3,5)$}} & FAD &  2.937  &2.687  &2.553  &2.583  &2.590  &7.114  &8.167  &9.238 &10.426 &11.157\\ 
   & EM & 47.200 &47.333 &49.867 &49.956 &71.308 &47.469 &47.119 &43.828 &45.304 &44.141\\ 
   \bottomrule 
\end{tabular}}
\end{table}

\begin{table}[ht]
\caption{Average CPU time (in seconds) of FAD and EM applied for data-driven models.} 
\label{table:app3}
\centering
\resizebox{.9\textwidth}{!}{\begin{tabular}{cc|cccccccc}
  \toprule 
  & & 1 & 2 & 3 & 4 & 5 & 6 & 7 & 8\\
 \midrule 
 \midrule\multirow{2}{*}{\bm{$(n,p,q)=(160,24547,2)$}} & FAD &  5.007  &4.222 &10.835 &13.636 &-- &-- &-- &--\\
 & EM & 253.021 &304.909 &311.916 &303.712 &-- &-- &-- &--\\ 
  \midrule\multirow{2}{*}{\bm{$(n,p,q)=(180,24547,2)$}} & FAD &  4.927  &4.104 &10.411 &12.058 &-- &-- &-- &--\\
 & EM & 287.824 &345.504 &331.919 &314.723 &-- &-- &-- &--\\ 
  \midrule\multirow{2}{*}{\bm{$(n,p,q)=(340,24547,4)$}} & FAD &  6.645  &7.121  &7.449  &6.688 &22.294 &26.575 &31.109 &34.208\\
 & EM & 648.759 &734.226 &745.902 &735.263 &767.010 &789.614 &802.502 &748.395\\ 
   \bottomrule 
\end{tabular}}
\end{table}

\subsection{Estimation errors in parameters}
\vspace{-0.25in}
\label{sec:supp-sim-RG}
\begin{figure}[H]
  \centering
  \mbox{
  \subfloat[]{\label{fig-fn-R-hd}\includegraphics[width = 0.31\textwidth]{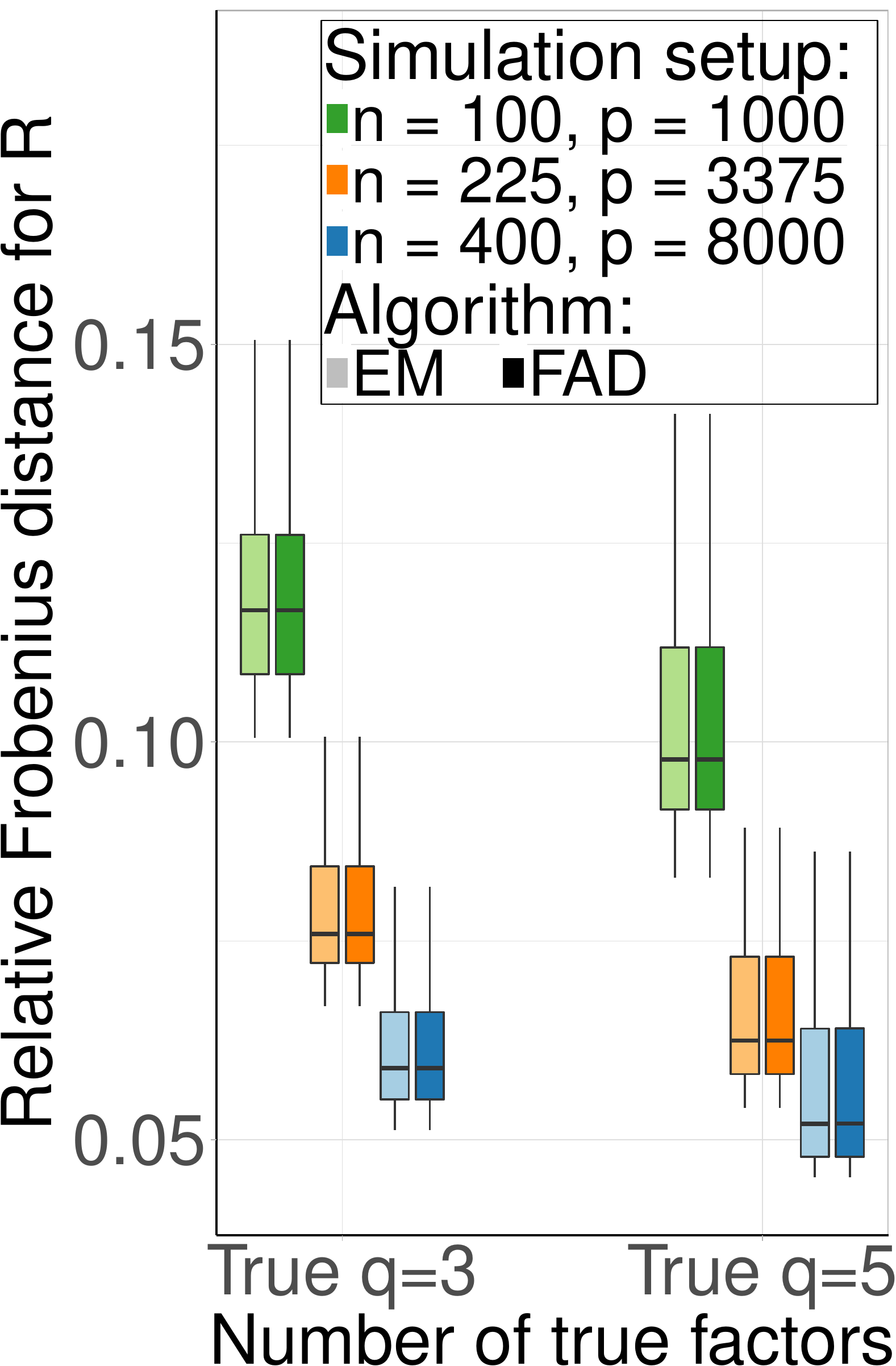}}
  \subfloat[]{\label{fig-fn-G-hd}\includegraphics[width = 0.3225\textwidth]{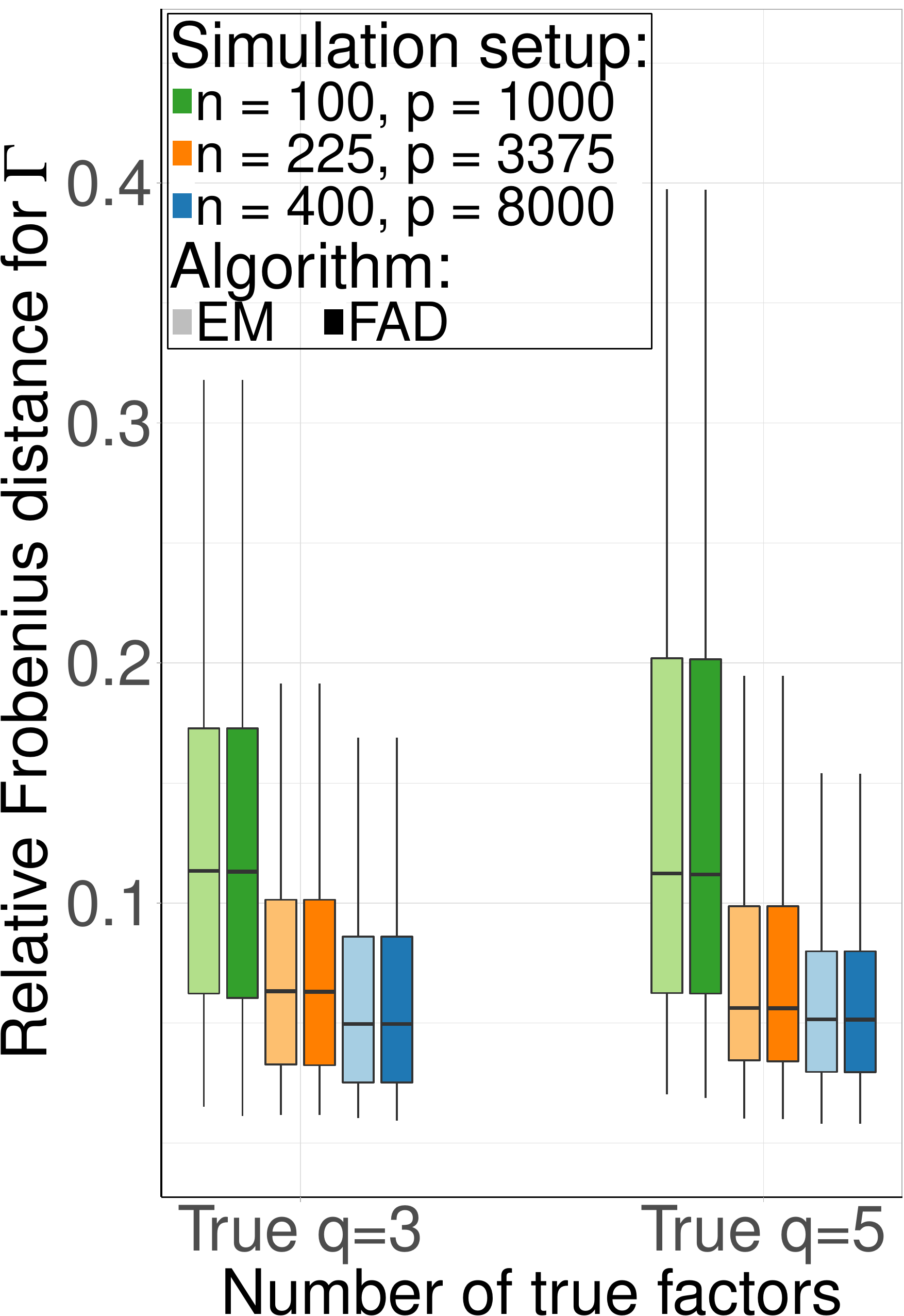}}
    \subfloat[]{\label{fig-fn-LL-hd}\includegraphics[width = 0.31\textwidth]{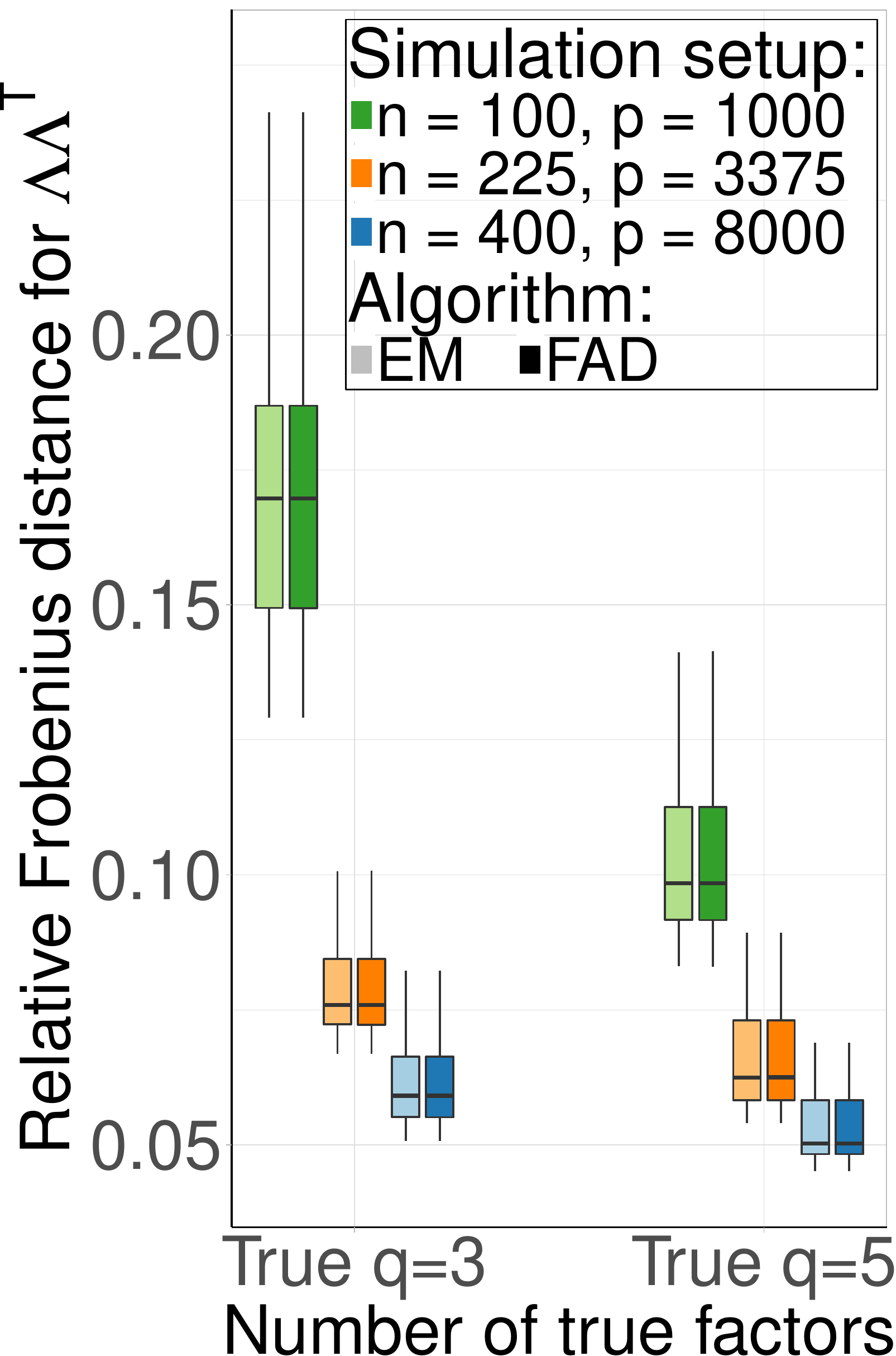}}}
\vspace{-0.15in}
\caption{Relative Frobenius errors of FAD and EM for (a) correlation matrix $\mathbf{R}$, (b) signal matrix $\bm\Gamma$ and (c) $\bm\Lambda\bm\Lambda^\top$ on randomly simulated cases, with lighter ones for EM and darker ones for FAD.}
\label{fig-sim-RG}
\end{figure}

\begin{figure}[H]
  \centering
\mbox{
    \subfloat[]{\label{fig-fn-R-uhd}\includegraphics[width = 0.315\textwidth]{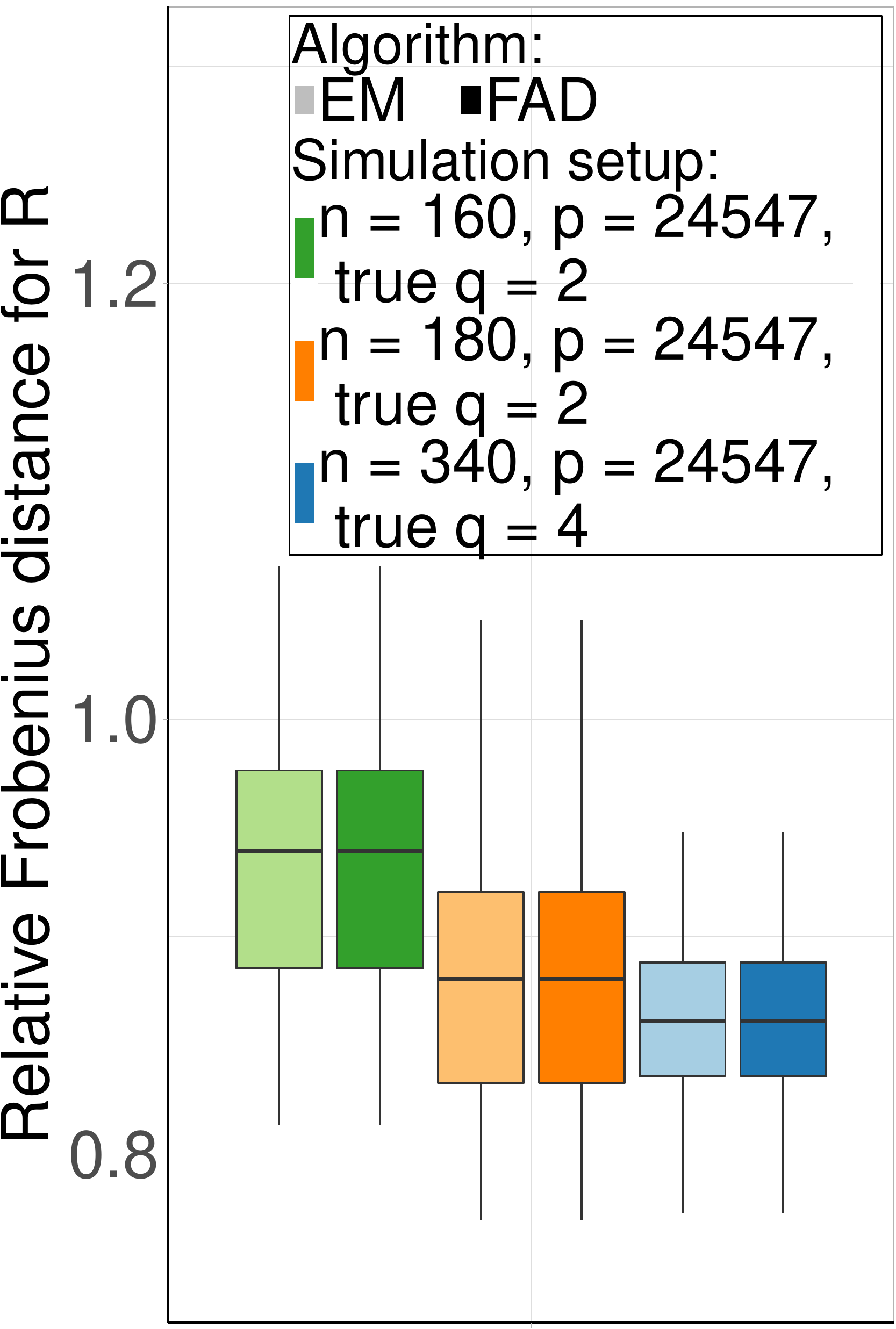}}
  \subfloat[]{\label{fig-fn-G-uhd}\includegraphics[width = 0.31\textwidth]{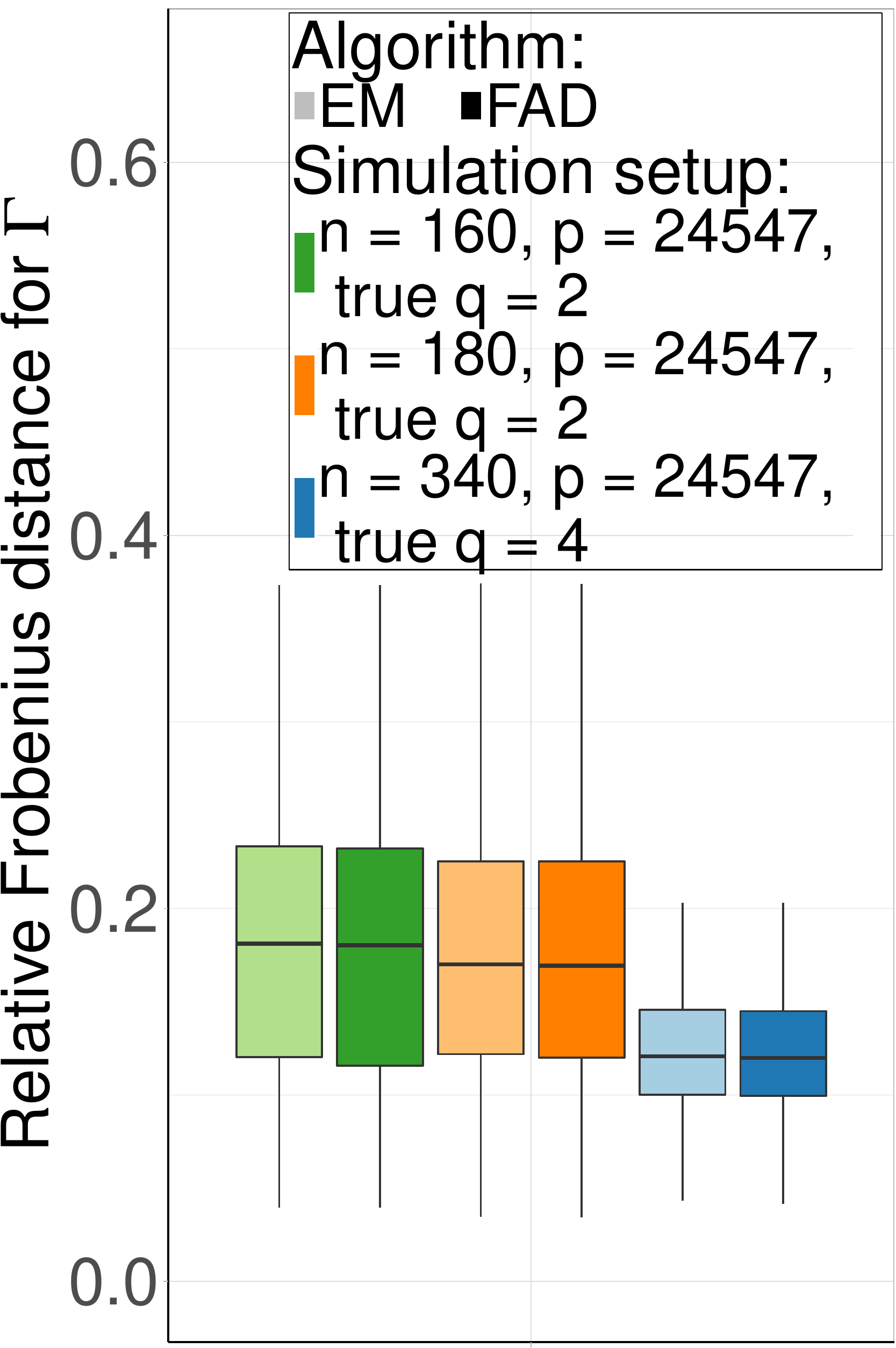}}
    \subfloat[]{\label{fig-fn-LL-hd}\includegraphics[width = 0.31\textwidth]{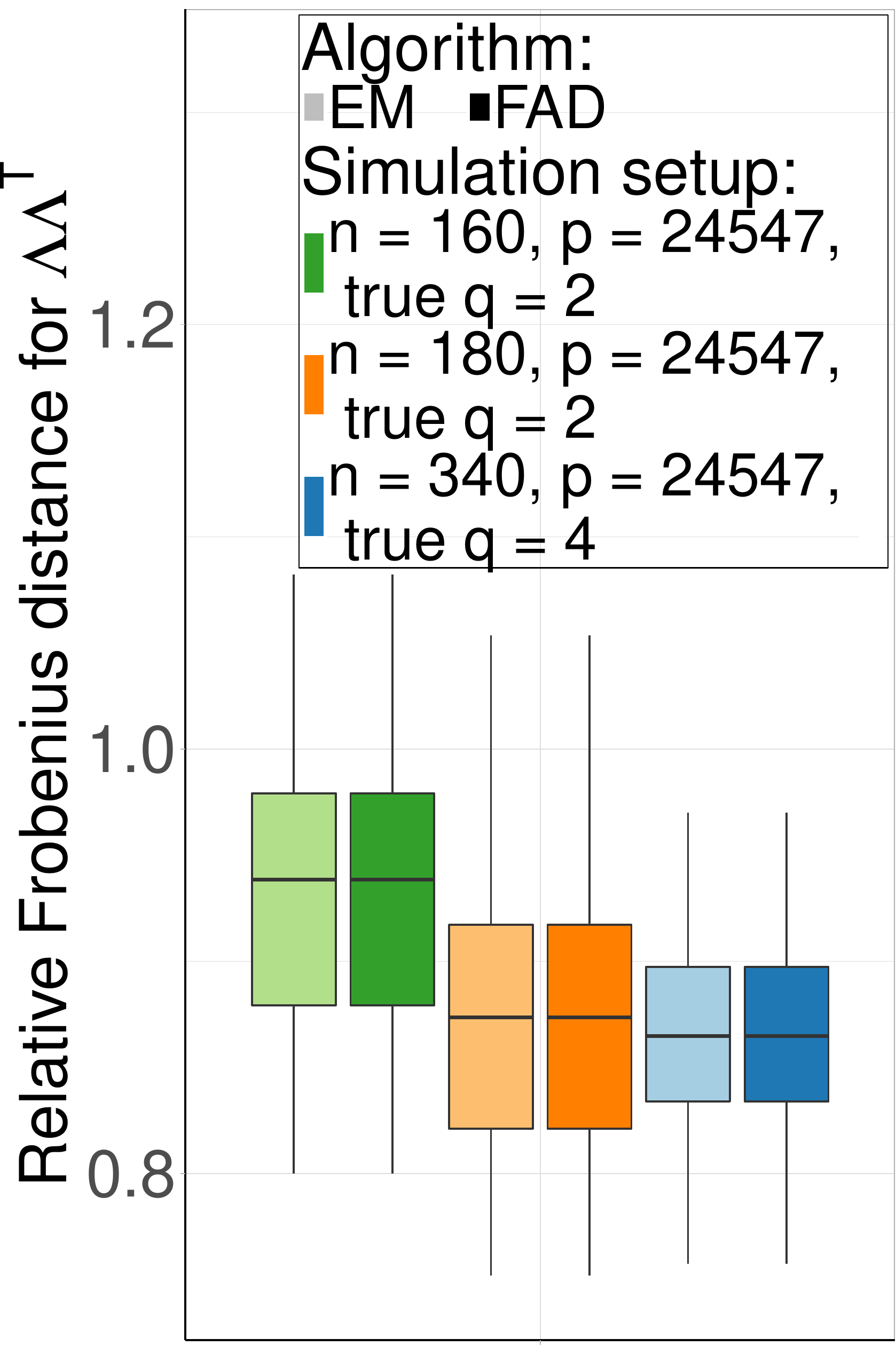}}
  }
  \vspace{-0.15in}
\caption{Relative Frobenius errors of FAD and EM for (a) correlation matrix $\mathbf{R}$, (b) signal matrix $\bm\Gamma$ and (c) $\bm\Lambda\bm\Lambda^\top$ on data-driven cases, with lighter ones for EM and darker ones for FAD.}
\label{fig-sim-RG}
\end{figure}
\vspace{-0.15in}
\subsection{ Performance of FAD compared to EM for high-noise models}
\vspace{-0.15in}
\label{sec:supp-sim-hnoise}
\begin{figure}[H]
\centering
\begin{minipage}{.34\textwidth}
  \subfloat[]{\label{fig-time-hn}\includegraphics[width = 0.99\textwidth]{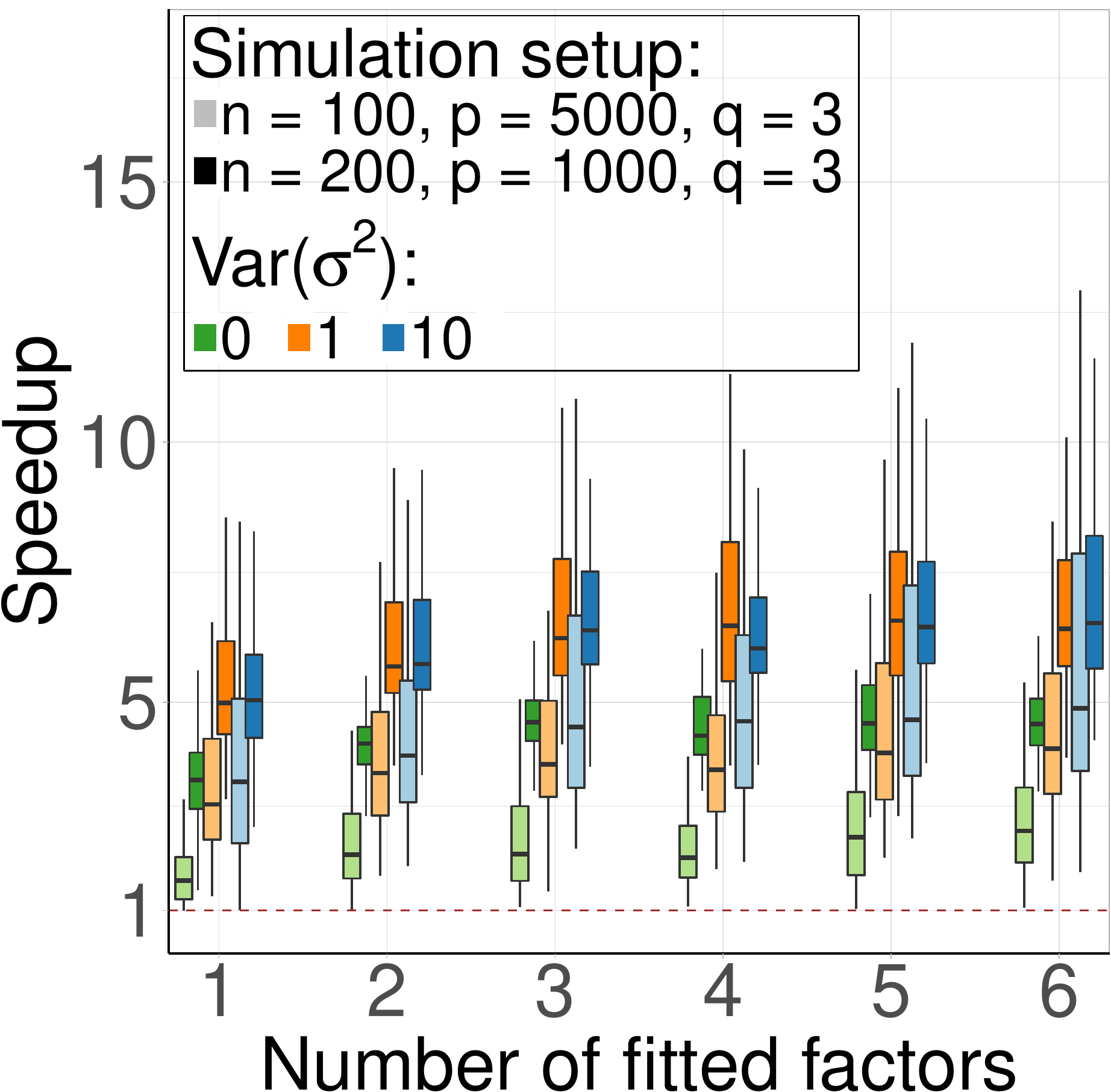}}
  \end{minipage}%
  \begin{minipage}{.227\textwidth}
  \subfloat[]{\label{fig-fn-R-hn}\includegraphics[width = 0.99\textwidth]{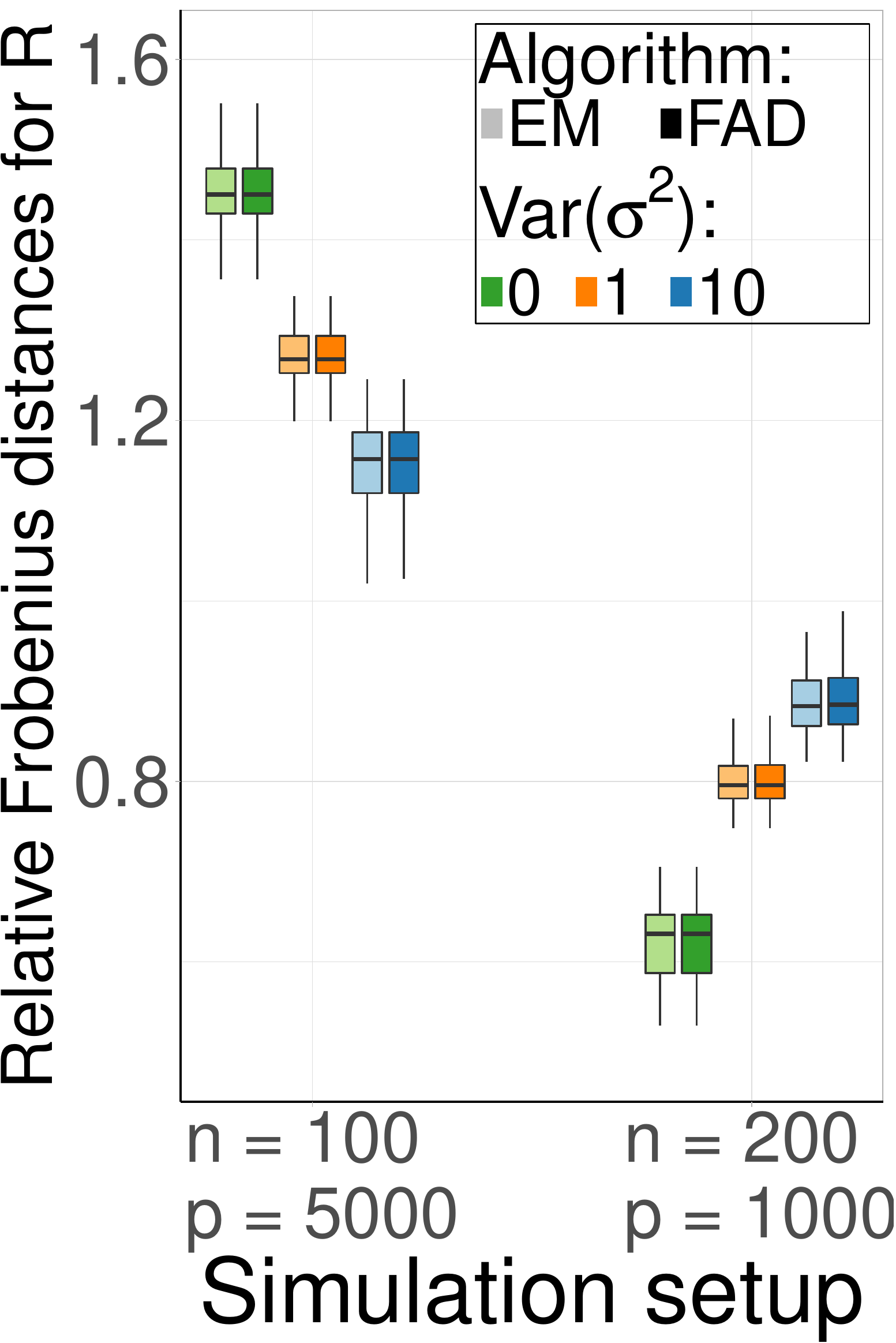}}
    \end{minipage}%
  \begin{minipage}{.215\textwidth}
  \subfloat[]{\label{fig-fn-G-hn}\includegraphics[width = 0.99\textwidth]{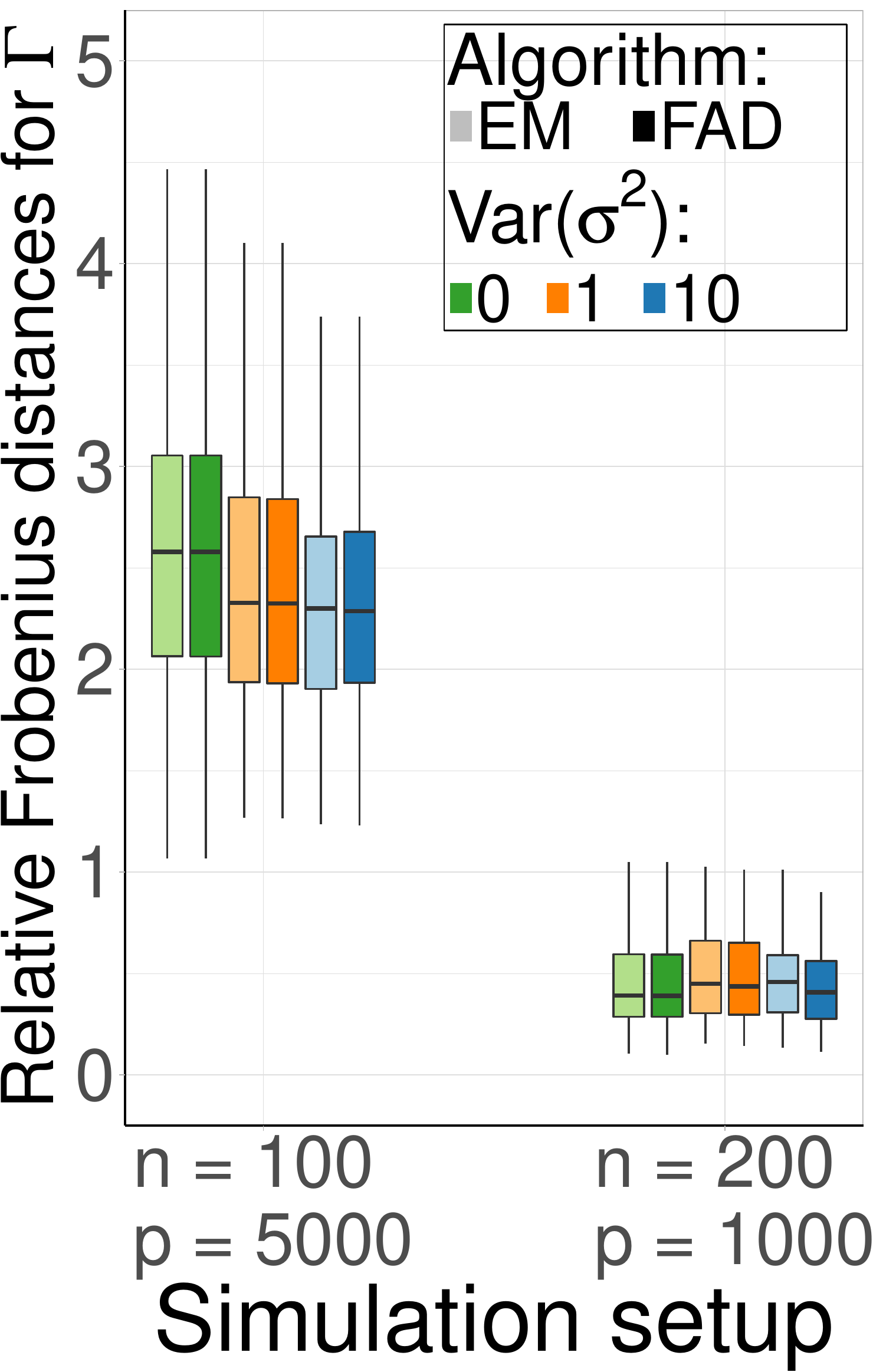}}
\end{minipage}%
  \begin{minipage}{.234\textwidth}
  \subfloat[]{\label{fig-fn-LL-hn}\includegraphics[width = 0.99\textwidth]{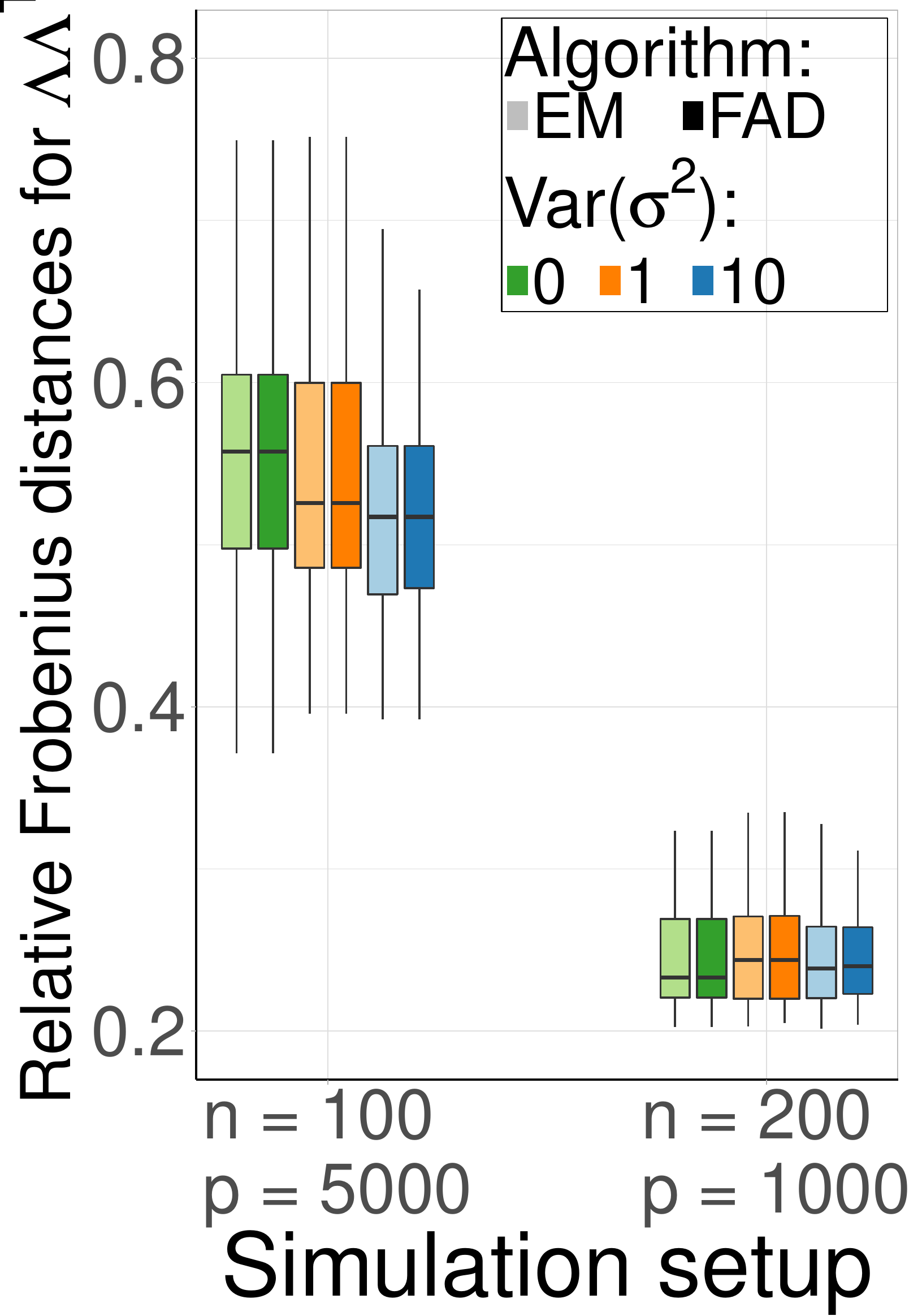}}
\end{minipage}
\vspace{-0.15in}
\caption{Relative speed of FAD to EM on dataset with high noise (a). Relative Frobenius errors of FAD and EM for (b) correlation matrix $\mathbf{R}$, (c) signal matrix $\bm\Gamma$ and (d) $\bm\Lambda\bm\Lambda^\top$.}
\label{fig-sim-hn}
\end{figure}

\bibliographystyle{IEEEtran}
\bibliography{acm}
\end{document}